\documentclass[12pt,letterpaper]{article}

\usepackage{natbib,hyperref}

\usepackage[ left=1in, top=1in, right=1in, bottom=1in]{geometry}
\usepackage{graphicx,bm,colonequals,amsmath,amssymb,url,xcolor,bbm}
\usepackage{array,tabularx,multirow}
\usepackage{enumitem}
\usepackage[font={footnotesize}]{caption,subcaption}
\usepackage[utf8]{inputenc}
\usepackage{enumitem}
\usepackage{soul} % for strikethrough text
\usepackage{placeins} % for FloatBarrier
\usepackage[normalem]{ulem}

\usepackage[ruled,vlined]{algorithm2e}
\SetKwInput{KwInput}{Input}
\usepackage{algorithmic}
\usepackage{array,framed}
\usepackage{float}
\makeatletter
\newcommand{\algorithmicbreak}{\textbf{break}}
\newcommand{\BREAK}{\STATE \algorithmicbreak}
\makeatother

\usepackage{mathtools}
\mathtoolsset{showonlyrefs} 

\bibpunct[, ]{(}{)}{;}{a}{,}{,}
   
\usepackage{amsthm}
\newtheoremstyle{propstyle} % name
    {2mm}                    % Space above
    {1mm}                    % Space below
    {\itshape}                   % Body font
    {}                           % Indent amount
    {\scshape}                   % Theorem head font
    {.}                          % Punctuation after theorem head
    {.5em}                       % Space after theorem head
    {}  % Theorem head spec (can be left empty, meaning ‘normal’)
\theoremstyle{propstyle}
\newtheorem{proposition}{Proposition}
\theoremstyle{propstyle}

\theoremstyle{propstyle}

\theoremstyle{propstyle}
\newtheorem{remark}{Remark}
\theoremstyle{propstyle}

\usepackage{array,framed}
\makeatletter
\renewcommand{\paragraph}{%
  \@startsection{paragraph}{4}%
  {\z@}{2ex \@plus 1ex \@minus .2ex}{-1em}%
  {\normalfont\normalsize\bfseries}%
}
\makeatother

\DeclareMathAlphabet\mathbfcal{OMS}{cmsy}{b}{n}

\newcommand{\bu}{\mathbf{u}}

\newcommand{\ba}{\mathbf{a}}
\newcommand{\bv}{\mathbf{v}}

\newcommand{\bs}{\mathbf{s}}

\newcommand{\bx}{\mathbf{x}}
\newcommand{\by}{\mathbf{y}}

\newcommand{\bt}{\mathbf{t}}

\newcommand{\bw}{\mathbf{w}}
\newcommand{\bS}{\mathbf{S}}

\newcommand{\bP}{\mathbf{P}}
\newcommand{\bA}{\mathbf{A}}

\newcommand{\bE}{\mathbf{E}}
\newcommand{\bL}{\mathbf{L}}

\newcommand{\bI}{\mathbf{I}}
\newcommand{\bD}{\mathbf{D}}
\newcommand{\bH}{\mathbf{H}}
\newcommand{\bU}{\mathbf{U}}

\newcommand{\bK}{\mathbf{K}}

\newcommand{\bQ}{\mathbf{Q}}
\newcommand{\bB}{\mathbf{B}}

\newcommand{\bR}{\mathbf{R}}

\newcommand{\bfzero}{\mathbf{0}}

\newcommand{\bfmu}{\bm{\mu}}

\newcommand{\bfSigma}{\bm{\Sigma}}

\newcommand{\bfLambda}{\bm{\Lambda}}

\DeclareMathOperator*{\var}{Var}
\DeclareMathOperator*{\cov}{Cov}
\DeclareMathOperator*{\diag}{diag}

\newcommand{\order}{\mathcal{O}}
\newcommand{\normal}{\mathcal{N}}
\newcommand{\norm}[1]{\lVert#1\rVert}

\DeclareFontFamily{OT1}{pzc}{}
\DeclareFontShape{OT1}{pzc}{m}{it}{<-> s * [1.2] pzcmi7t}{}
\DeclareMathAlphabet{\mathpzc}{OT1}{pzc}{m}{it}
\newcommand{\sx}{\mathcal{X}}
\newcommand{\sa}{\mathcal{A}}

\newcommand{\condset}{\mathcal{C}}

\newcommand{\domain}{\mathcal{D}}
\newcommand{\locs}{\mathcal{S}}

\newcommand{\im}{{i_1,\ldots,i_m}}

\newcommand{\il}{{i_1,\ldots,i_l}}
\newcommand{\js}{{j_1,\ldots,j_s}}
\newcommand{\jm}{{j_1,\ldots,j_m}}
\newcommand{\jmp}{{j_1,\ldots,j_{m+1}}}
\newcommand{\jmm}{{j_1,\ldots,j_{m-1}}}

\newcommand{\jk}{{j_1,\ldots,j_k}}
\newcommand{\jkp}{{j_1,\ldots,j_{k+1}}}

\newcommand{\jl}{{j_1,\ldots,j_\ell}}

\newcommand{\jlp}{{j_1,\ldots,j_{\ell+1}}}
\newcommand{\jM}{{j_1,\ldots,j_M}}

\newcommand{\evol}{\mathcal{E}}
\newcommand{\levol}{\mathbf{E}}
\newcommand{\grid}{\mathcal{S}}
\newcommand{\obs}{\mathcal{I}}

\DeclareMathOperator*{\rchol}{rchol}
\DeclareMathOperator*{\chol}{chol}

\DeclareMathOperator*{\ichol}{\text{ichol}}
\DeclareMathOperator*{\hv}{\texttt{HV}}
\DeclareMathOperator*{\hvl}{\texttt{HVL}}

\newcommand{\condsetN}{N}

\newcommand{\An}{\text{An}}

\title{Hierarchical sparse Cholesky decomposition with applications to high-dimensional spatio-temporal filtering}

\author{Marcin Jurek\thanks{Department of Statistics and Data Science, University of Texas at Austin} \and Matthias Katzfuss\thanks{Department of Statistics, Texas A\&M University. Corresponding author: \texttt{katzfuss@gmail.com}.}}
\date{}

\begin{document}

\maketitle

\begin{abstract}
Spatial statistics often involves Cholesky decomposition of covariance matrices. To ensure scalability to high dimensions, several recent approximations have assumed a sparse Cholesky factor of the precision matrix. We propose a hierarchical Vecchia approximation, whose conditional-independence assumptions imply sparsity in the Cholesky factors of both the precision and the covariance matrix. This remarkable property is crucial for applications to high-dimensional spatio-temporal filtering. We present a fast and simple algorithm to compute our hierarchical Vecchia approximation, and we provide extensions to non-linear data assimilation with non-Gaussian data based on the Laplace approximation. In several numerical comparisons, including a filtering analysis of satellite data, our methods strongly outperformed alternative approaches.
\end{abstract}


%%%%%%%%%%%%%%%%%%%%%%%%%%%%%%%%%%%%%%%%%%%%%%%%%%%%%%%
\section{Introduction}
\label{intro}

%%% motivation
Symmetric positive-definite matrices arise in spatial statistics, Gaussian-process inference, and spatio-temporal filtering, with a wealth of application areas, including geoscience \citep[e.g.,][]{Cressie1993,Banerjee2004}, machine learning \citep[e.g.,][]{Rasmussen2006}, data assimilation \citep[e.g.,][]{Nychka2010,Katzfuss2015b}, and the analysis of computer experiments \citep[e.g.,][]{Sacks1989,Kennedy2001}.
Inference in these areas typically relies on Cholesky decomposition of the positive-definite matrices. However, this operation scales cubically in the dimension of the matrix, and it is thus computationally infeasible for many modern problems and applications, which are increasingly high-dimensional.
% The advantages of Cholesky factorization make it a ubiquitous tool in filtering problems and give rise to the family of square root filters (reference).

%%%% Literature review

%% general
Countless approaches have been proposed to address these computational challenges. \citet{Heaton2017} provide a recent review from a spatial-statistics perspective, and \citet{Liu2018} review approaches in machine learning. In  high-dimensional filtering, proposed solutions include low-dimensional approximations \citep[e.g.,][]{Verlaan1995,Pham1998,Wikle1999,Katzfuss2010}, spectral methods \citep[e.g.][]{Wikle1999,Sigrist2015}, and hierarchical approaches \citep[e.g.,][]{Johannesson2003,LiAmbikasaran2014,Saibaba2015,Jurek2018}. Operational data assimilation often relies on ensemble Kalman filters \citep[e.g.,][]{Evensen1994,Burgers1998,Anderson2001,Evensen2007,Katzfuss2015b,Katzfuss2017c}, which represent distributions by samples or ensembles. %; this reduces the computational burden, but it also results in undesired sampling variability. 

%% vecchia and sparse cholesky
Perhaps the most promising approximations for spatial data and Gaussian processes implicitly or explicitly rely on sparse Cholesky factors. The assumption of ordered conditional independence in the popular Vecchia approximation \citep{Vecchia1988} and its extensions \citep[e.g.,][]{Stein2004,Datta2016,Guinness2016a,Katzfuss2017a,Katzfuss2018,Katzfuss2020,Schafer2020} implies sparsity in the Cholesky factor of the precision matrix. \citet{Schafer2017} uses an incomplete Cholesky decomposition to construct a sparse approximate Cholesky factor of the covariance matrix. However, these methods are not generally applicable to spatio-temporal filtering, because the assumed sparsity is not preserved under filtering operations.

%%%% Our approach 

%% spatial-only case
Here, we relate the sparsity of the Cholesky factors of the covariance matrix and the precision matrix to specific assumptions regarding ordered conditional independence.
%While some of these conditions have been known before, we present them in a unified framework and provide precise proofs.
We show that these assumptions are simultaneously satisfied for a particular Gaussian-process approximation that we call hierarchical Vecchia (HV), which is a special case of the general Vecchia approximation \citep{Katzfuss2017a} based on hierarchical domain partitioning \citep[e.g.,][]{Katzfuss2015,Katzfuss2017b}.
We show that the HV approximation can be computed using a simple and fast incomplete Cholesky decomposition.
%which is conceptually simpler and more accurate than existing algorithms for multi-resolution decompositions. 

%% filtering
Due to its remarkable property of implying a sparse Cholesky factor whose inverse has equivalent sparsity structure, HV is well suited for extensions to spatio-temporal filtering; this is in contrast to other Vecchia approximations and other spatial approximations relying on sparsity. We provide a scalable HV-based filter for linear Gaussian spatio-temporal state-space models, which is related to the multi-resolution filter of \citet{Jurek2018}.
Further, by combining HV with a Laplace approximation \citep[cf.][]{Zilber2019}, our method can be used for the analysis of non-Gaussian data. Finally, by combining the methods with the extended Kalman filter \citep[e.g.,][Ch.~5]{Grewal1993}, we obtain fast filters for high-dimensional, non-linear, and non-Gaussian spatio-temporal models.
For a given formulation of HV, the computational cost of all of our algorithms scales linearly in the state dimension, assuming sufficiently sparse temporal evolution.

% \todo{should we include a paragraph stating the contributions of this paper?}
This paper makes several important contributions. First, it succinctly summarizes and proves conditional-independence conditions that ensure that particular elements of the Cholesky factor and its inverse vanish; while some of these conditions have been known before, we present them in a unified framework and provide precise proofs. Second, we describe a new version of the Vecchia approximation and demonstrate that it satisfies all of the conditional-independence conditions. Third, we show how the HV approximation can be easily computed using an incomplete Cholesky decomposition (IC0). Fourth, we present how IC0 can be used to create an approximate Kalman filter. Finally, we describe how all these developments can be combined in a filtering algorithm applicable to non-Gaussian data and non-linear temporal evolution models.

%%%%%  organization of the manuscript
The remainder of this document is organized as follows. In Section \ref{sec:sparsity}, we specify the relationship between ordered conditional independence and sparse (inverse) Cholesky factors. Then, we build up increasingly complex and general methods, culminating in non-linear and non-Gaussian spatio-temporal filters: in Section \ref{sec:hv}, we introduce the HV approximation for a linear Gaussian spatial field at a single time point; in Section \ref{sec:nongaussian}, we extend this to non-Gaussian data; and in Section \ref{sec:filter}, we consider the general spatio-temporal filtering case, including nonlinear evolution. Section \ref{sec:comparison} contains numerical comparisons to existing approaches. Section \ref{sec:data-application} presents a filtering analysis of satellite data. Section \ref{sec:conclusion} concludes.
Appendices \ref{app:graph}--\ref{app:proofs} contain proofs and further details.
Sections S1-S6 in the Supplementary Material provide additional information, including on a particle filter for parameter inference on unknown parameters in the model.
Code implementing our methods and numerical comparisons is available at \url{https://github.com/katzfuss-group/vecchiaFilter}.

%%%%%%%%%%%%%%%%%%%%%%%%%%%%%%%%%%%%%%%%%%%%%%%%%%%%%%%
\section{Sparsity of Cholesky factors\label{sec:sparsity}}

We begin by specifying the connections between ordered conditional independence and sparsity of the Cholesky factor of the covariance and precision matrix.
\begin{remark}
Let $\bw$ be a normal random vector with variance-covariance matrix $\bK$.
\begin{enumerate}
    \item Let $\bL = \chol(\bK)$ be the lower-triangular Cholesky factor of the covariance matrix $\bK$. For $i>j$:
    \[ \bL_{i,j}=0 \iff w_i \perp w_j \, | \, \bw_{1:j-1} \]
    \item Let $\bU = \rchol(\bK^{-1}) = \bP \chol( \bP\bK^{-1}\bP)\,\bP$ be the Cholesky factor of the precision matrix under reverse ordering, where $\bP$ is the reverse-ordering permutation matrix. Then $\bU$ is upper-triangular, and for $i>j$:
    \[ \bU_{j,i}=0 \iff w_i \perp w_j \, | \, \bw_{1:j-1},\bw_{j+1:i-1} \]
\end{enumerate}
\label{remark:cond-indep}
\end{remark}
The connection between ordered conditional independence and the Cholesky factor of the precision matrix is well known \citep[e.g.,][]{Rue2010}; Part 2 of our remark states this connection under reverse ordering \citep[e.g.,][Prop.~3.3]{Katzfuss2017a}. In Part 1, we consider the lesser-known relationship between ordered conditional independence and sparsity of the Cholesky factor of the covariance matrix, which was recently discussed in \citet[][Sect.~1.4.2]{Schafer2017}. For completeness, we provide a proof of Remark \ref{remark:cond-indep} in Appendix \ref{app:proofs}.

Remark \ref{remark:cond-indep} is crucial for our later developments and proofs. In Section \ref{sec:hv}, we specify the HV approximation of Gaussian processes that satisfies both types of conditional independence in Remark \ref{remark:cond-indep}; the resulting sparsity of the Cholesky factor and its inverse allows extensions to spatio-temporal filtering in Section \ref{sec:filter}.

%%%%%%%%%%%%%%%%%%%%%%%%%%%%%%%%%%%%%%%%%%%%%%%%%%%%%%%
\section{HV approximation for large Gaussian spatial data\label{sec:hv}}

Consider a Gaussian process $x(\cdot)$ and a vector $\bx = (x_1, \ldots, x_n)^\top$ representing $x(\cdot)$ evaluated on a grid $\grid = \{\bs_1,\ldots,\bs_n \}$ with $\bs_i \in \domain \subset \mathbb{R}^d$ and $x_i = x(\bs_i)$ for $\bs_i \in \domain$, $i=1,\ldots,n$. We assume the following model:
\begin{align}
\label{eq:gobs} y_{i} \,|\, \bx & \stackrel{ind}{\sim} \normal(x_i,\tau_i^2),\qquad i \in \obs, \\
\label{eq:gevol} \bx  & \sim \normal_n(\bfmu,\bfSigma),
\end{align}
where $\mathcal{N}$ and $\mathcal{N}_n$ denote, respectively, a univariate and multivariate normal distribution, $\obs \subset \{1,\ldots,n\}$ are the indices of grid points at which observations are available, and $\by$ is the data vector consisting of these observations $\{y_{i}: i \in \obs\}$.
Note that we can equivalently express \eqref{eq:gobs} using matrix notation as $\by \,|\, \bx \sim \normal(\bH\bx,\bR)$, where $\bH$ is obtained by selecting only the rows with indices $i \in \obs$ from an identity matrix, and $\bR$ is a diagonal matrix with entries $\{\tau_i^2: i \in \obs\}$. While \eqref{eq:gobs} assumes conditional independence of each observation, we discuss in Section S1 how this assumption can be relaxed.

Our interest is in computing the posterior distribution of $\bx$ given $\by$, which requires inverting or decomposing an $n \times n$ matrix at a cost of $\order(n^3)$ if $|\obs| = \order(n)$. This is computationally infeasible for large $n$.

%%%%%%
\subsection{The HV approximation\label{sec:hv-intro}}

We now describe the HV approximation with unique sparsity and computational properties, which enable fast computation for spatial models as in \eqref{eq:gobs}--\eqref{eq:gevol} and also allow extensions to spatio-temporal filtering as explained later.

Assume that the elements of the vector $\bx$ are hierarchically partitioned into sets  $\sx^m$, for $m=0, \dots M$, where for $m>1$ we have $\sx^m = \bigcup_{j_1=1}^{J_1} \dots \bigcup_{j_m=1}^{J_m} \sx_\jm$, and $\sx_\jm$ is a set consisting of $|\sx_\jm|$ elements of $\bx$, such that there is no overlap between any two sets, $\sx_\jm \cap \sx_\il = \emptyset$ for $(\jm) \neq (\il)$. Note that this also means that the sets $\sx^m$ and $\sx^l$ are disjoint for $m \neq l$. Define $\sx^{0:m} = \bigcup_{k=0}^m \sx^k$. 
We assume that $\bx$ is ordered according to $\sx^{0:M}$, in the sense that if $i>j$, then $x_i \in \sx^{m_1}$ and $x_j \in \sx^{m_2}$ with $m_1 \geq m_2$. We also say that if $m>l$, then elements of $\sx^m$ are at a higher level than elements of $\sx^l$.
As a toy example with $n=6$, the vector $\bx = (x_1,\ldots,x_6)$ might be partitioned with $M=1$, $J_1 = 2$, as $\sx^{0:1} = \sx^0 \cup \sx^1$, $\sx^0 = \sx = \{x_1,x_2\}$, and $\sx^1 = \sx_{1,1} \cup \sx_{1,2}$, where $\sx_{1,1}=\{x_3,x_4\}$, and $\sx_{1,2} = \{x_5,x_6\}$. Another toy example is illustrated in Figure \ref{fig:toy}.

The exact distribution of $\bx \sim \normal_n(\bfmu,\bfSigma)$ can be written as
\[
\textstyle p(\bx) = \prod_{m=0}^M \prod_\jm p(\sx_\jm|\sx^{0:m-1}, \sx_{\jmm,1:j_m-1}),
\]
where the conditioning set of $\sx_\jm$ consists of all sets $\sx^{0:m-1}$ at lower levels, plus those at the same level that are previous in lexicographic ordering. The idea of \citet{Vecchia1988} was to remove many of these variables in the conditioning set, which for geostatistical applications often incurs only small approximation error due to the so-called screening effect \citep[e.g.,][]{Stein2002,Stein2011}.

\begin{figure}[!htbp]
    \centering
    \begin{subfigure}{.99\textwidth}
    \centering
    \includegraphics[width=0.9\textwidth]{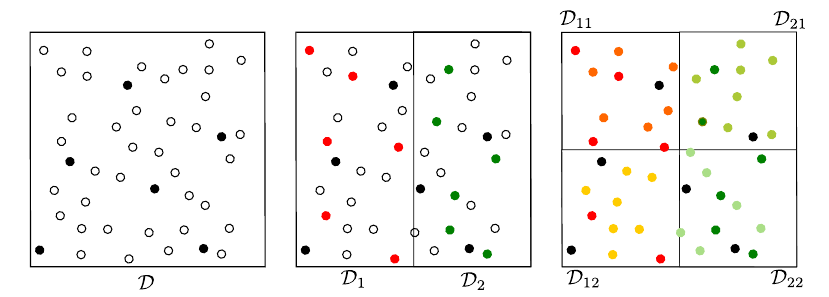}
    \caption{Iterative domain partitioning for $n=35$ locations} 
    \label{fig:domain-&-knots}
\end{subfigure}\\
\vspace{3mm}
    \begin{subfigure}{0.54\textwidth}
    \includegraphics[width=1.0\textwidth]{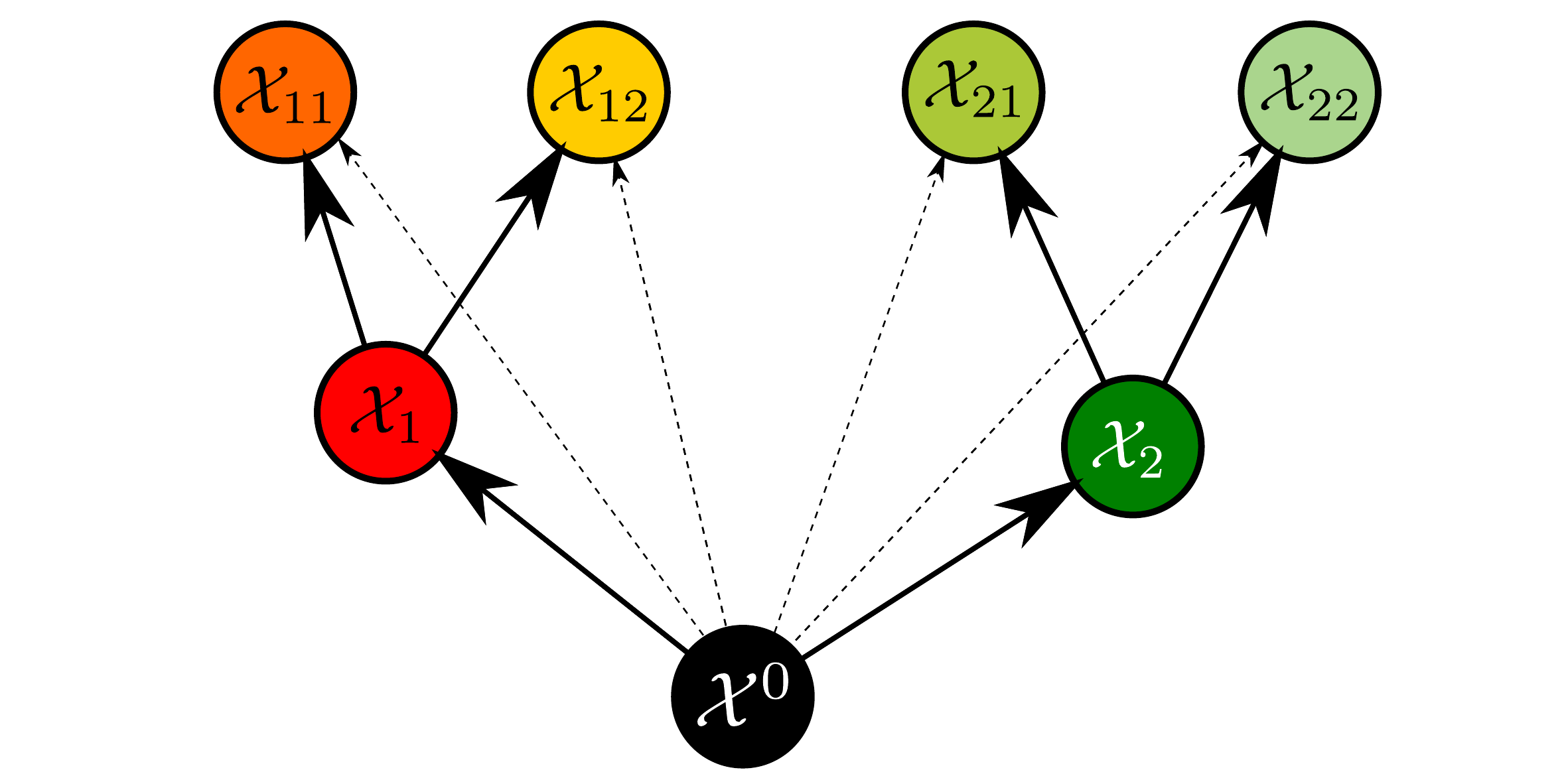}
    \caption{Directed acyclic graph (DAG)}
    \label{fig:dag}
    \end{subfigure}
    \begin{subfigure}{.44\textwidth}
    \includegraphics[trim=0mm 12mm 0mm 12mm, clip,width=0.95\textwidth]{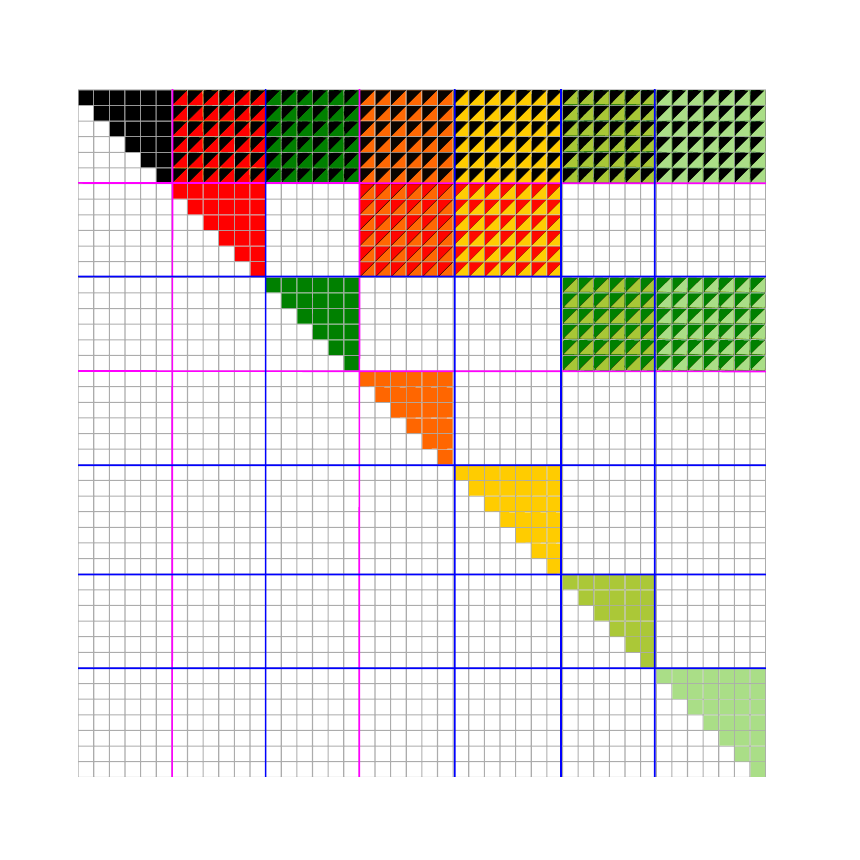}
    \caption{Sparsity pattern of the $\bU$ matrix}
    \label{fig:sparsity-pattern}
    \end{subfigure}
    \caption{Toy example with $n=35$ of the HV approximation in \eqref{eq:vecchia-approx} with $M=2$ and $J_1=J_2=2$; the color for each set $\sx_\jm$ is consistent across (a)--(c).
    (a) Partitioning of the spatial domain $\domain$ and the locations $\locs$; for level $m=0,1,2$, locations of $\sx^{0:m}$ (solid dots) and locations of points at higher levels ($\circ$). 
    (b) DAG illustrating the conditional-dependence structure, with  bigger arrows for connections between vertices at neighboring levels of the hierarchy, to emphasize the tree structure.
    (c) Corresponding sparsity pattern of $\bU$ (see Proposition \ref{prop:same-sparsity-factors}), with groups of columns/rows corresponding to different levels separated by pink lines, and groups of columns/rows corresponding to different $\sx_\jm$ at the same level separated by blue lines. }
    \label{fig:toy}
\end{figure}

Here we consider the HV approximation of the form
\begin{equation}
\textstyle \hat{p}(\bx) = \prod_{m=0}^M \prod_\jm p(\sx_\jm|\sa_\jm),
\label{eq:vecchia-approx}
\end{equation}
where $\sa_\jm = \sx \cup \sx_{j_1} \cup \ldots \cup \sx_{\jmm}$. We call $\sa_\jm$ the set of ancestors of $\sx_\jm$. For example, the set of ancestors of $\sx_{2,1,2}$ is $\sa_{2,1,2} = \sx \cup \sx_2 \cup \sx_{2,1}$. Thus, $\sa_\jm = \sa_\jmm \cup \sx_\jmm$, and the ancestor sets are nested: $\sa_\jmm \subset \sa_\jm$. We can equivalently write \eqref{eq:vecchia-approx} in terms of individual variables as
\begin{equation}
\textstyle \hat{p}(\bx) = \prod_{i=1}^n p(x_i|\,\condset_i),
\label{eq:vecchia-approx-ind}
\end{equation}
where $\condset_i = \sa_\jm \cup \{x_k \in \sx_\jm \! : \, k<i \}$ for $x_i \in \sx_\jm$. The choice of the $\condset_i$ involves a trade-off: generally, the larger the $\condset_i$, the higher the computational cost (see Proposition \ref{prop:complexity-vf} below), but the smaller the approximation error; HV is exact when all $\condset_i = \{x_1,\ldots,x_{i-1}\}$. 
% Note that $\mathbb{E}_{\hat{p}}[\bx] = \mathbb{E}_p[\bx] = \bfmu$, which can be shown by straightforward extension of the proof of \citet[][Prop.~1]{Katzfuss2017a}.

Vecchia approximations and their conditional-independence assumptions are closely connected to directed acyclic graphs \citep[DAGs;][]{Datta2016,Katzfuss2017a}. Summarizing briefly, as illustrated in Figure \ref{fig:dag}, we associate a vertex with each set $\sx_\jm$, and we draw an arrow from the vertex corresponding to $\sx_{\il}$ to the vertex corresponding to $\sx_\jm$ if and only if $\sx_{\il}$ is in the conditioning set of $\sx_\jm$ (i.e., $\sx_{\il} \subset \sa_\jm$). DAGs corresponding to HV approximations always have a tree structure, due to the nested ancestor sets. Necessary terminology and notation from graph theory is reviewed in Appendix \ref{app:graph}.

In practice, as illustrated in Figure \ref{fig:domain-&-knots}, we partition the spatial field $\bx$ into the hierarchical set $\sx^{0:M}$ based on a recursive partitioning of the spatial domain $\domain$ into $J_1$ regions $\domain_{1},\ldots,\domain_{J_1}$, each of which is again split into $J_2$ regions, and so forth, up to level $M$ \citep{Katzfuss2015}: $\domain_\jmm = \bigcup_{j_m=1}^{J_m} \domain_\jm$, $m=1,\ldots,M$.
We then set each $\sx_\jm$ to be a subset of the variables in $\bx$ whose location is in $\domain_\jm$: $\sx_\jm \subset \{x_i: \bs_i \in \domain_\jm\}$. This implies that the ancestors $\sa_\jm$ of each set $\sx_\jm$ consist of the variables associated with regions at lower levels $m=0,\ldots,m-1$ that contain $\domain_\jm$.
Specifically, for all our numerical examples, we set $J_1=\ldots=J_M=2$, and we select each $\sx_\jm$ corresponding to the first $|\sx_\jm|$ locations in a maximum-distance ordering \citep{Guinness2016a,Schafer2017} of $\locs$ that are contained in $\domain_\jm$ but are not already in $\sa_\jm$.

The HV approximation \eqref{eq:vecchia-approx} is closely related to the multi-resolution approximation \citep[MRA;][]{Katzfuss2015,Katzfuss2017b}, as noted in \citet[][Sec.~2.5]{Katzfuss2017a}; specifically, while HV makes conditional-independence assumptions that result in an approximate covariance matrix $\hat\bfSigma$ of $\bx$ with a sparse Cholesky factor (see Proposition \ref{prop:same-sparsity-factors} below), the MRA relies on a basis-function representation of a spatial process that results in a sparse nontriangular matrix square root. However, the approximate covariance matrices $\hat\bfSigma$ implied by both HV and MRA are hierarchical off-diagonal low-rank (HODLR) matrices \citep[e.g.][]{Hackbusch2015,Ambikasaran2016,Saibaba2015,Geoga2018}, as was noted for the MRA in \citet{Jurek2018}. The definition, exposition, and details based on conditional independence and sparse Cholesky factors provided here enable our later proofs, simple incomplete-Cholesky-based computation, and extensions to non-Gaussian data and to nonlinear space-time filtering.

%%%%%%
\subsection{Sparsity of the HV approximation}

For all Vecchia approximations, the assumed conditional independence implies a sparse Cholesky factor of the precision matrix \citep[e.g.,][Prop.~3.3]{Datta2016,Katzfuss2017a}. The conditional-independence assumption made in our HV approximation also implies a sparse Cholesky factor of the covariance matrix, which is in contrast to many other formulations of the Vecchia approximation. Let $\mathcal{N}(\bx|\bfmu, \bfSigma)$ denote the density of a normal distribution with mean $\bfmu$ and covariance matrix $\bfSigma$ evaluated at $\bx$.
\begin{proposition}
For the HV approximation in \eqref{eq:vecchia-approx}, we have $\hat p(\bx) = \normal_n(\bx|\bfmu,\hat\bfSigma)$. Define $\bL = \chol(\hat\bfSigma)$ and $\bU = \rchol(\hat\bfSigma^{-1}) = \bP \chol( \bP\hat\bfSigma^{-1}\bP)\,\bP$, where $\bP$ is the reverse-ordering permutation matrix.
\begin{enumerate}
    \item For $i\neq j$:
    \begin{enumerate}
        \item $\bL_{i,j} = 0$ unless $x_j \in \condset_i$
        \item $\bU_{j,i} = 0$ unless $x_j \in \condset_i$
    \end{enumerate}
    \item $\bU = \bL^{-\top}$
\end{enumerate}
\label{prop:same-sparsity-factors}
\end{proposition}
The proof relies on Remark \ref{remark:cond-indep}. All proofs can be found in Appendix \ref{app:proofs}. Proposition \ref{prop:same-sparsity-factors} says that the Cholesky factors of the covariance and precision matrix implied by a HV approximation are both sparse, and $\bU$ has the same sparsity pattern as $\bL^{\top}$. An example of this pattern is shown in Figure \ref{fig:sparsity-pattern}. Furthermore, because $\bL = \bU^{-\top}$, we can quickly compute one of these factors given the other, as described in Section \ref{sec:vecchia-complexity} below (see the proof of Proposition \ref{prop:complexity-vf}).

For other Vecchia approximations, the sparsity of the prior Cholesky factor $\bU$ for $\bx$ does not necessarily imply the same sparsity for the Cholesky factor of the posterior precision matrix of $\bx$ given $\by$, and in fact there can be substantial fill-in \citep{Katzfuss2017a}. However, this is not the case for the particular case of HV, for which the posterior sparsity is exactly the same as the prior sparsity:
\begin{proposition}
Assume that $\bx$ has the distribution $\hat p(\bx)$ given by the HV approximation in \eqref{eq:vecchia-approx}. Let $\widetilde\bfSigma = \var(\bx|\by)$ be the posterior covariance matrix of $\bx$ given data $y_{i} \,|\, \bx \stackrel{ind}{\sim} \normal(x_i,\tau_i^2)$, $i \in \obs \subset \{1,\ldots,n\}$, as in \eqref{eq:gobs}. Then:
\begin{enumerate}
    \item $\widetilde\bU = \rchol(\widetilde\bfSigma^{-1})$ has the same sparsity pattern as $\bU = \rchol(\hat\bfSigma^{-1})$. 
    \item $\widetilde\bL = \chol(\widetilde\bfSigma)$ has the same sparsity pattern as $\bL = \chol(\hat\bfSigma)$.
\end{enumerate}
\label{prop:posterior}
\end{proposition}

%%%%%%
\subsection{Fast computation using incomplete Cholesky factorization\label{sec:vecchia-complexity}}

For notational and computational convenience, we assume now that each conditioning set $\condset_i$ consists of at most $\condsetN$ elements of $\bx$. For example, this can be achieved by setting $|\sx_\jm|\leq r$ with $r = \condsetN/(M+1)$. Then $\bU$ can be computed using general expressions for the Vecchia approximation in $\order(n\condsetN^3)$ time \citep[e.g.,][]{Katzfuss2017a}. Inference using the related multi-resolution decomposition \citep{Katzfuss2015,Katzfuss2017b,Jurek2018} can be carried out in $\order(n\condsetN^2)$ time, but these algorithms are fairly involved. 

Instead, we show here how HV inference can be carried out in $\order(n\condsetN^2)$ time using standard sparse-matrix algorithms, including the incomplete Cholesky factorization, based on at most $n\condsetN$ entries of $\bfSigma$. Our algorithm, which is based on ideas in \citet{Schafer2017}, is much simpler than multi-resolution decompositions.

\begin{algorithm}[ht]
\caption{Incomplete Cholesky decomposition: $\ichol(\bA,\bS)$}
\KwInput{ positive-definite matrix $\bA \in \mathbb{R}^{n \times n}$, sparsity matrix $\bS \in \{0,1\}^{n \times n}$}
\KwResult{ lower-triangular $n\times n$ matrix $\bL$ }
\begin{algorithmic}[1]
\FOR{$i=1$ to $n$}
\FOR{$j=1$ to $i-1$}
\IF{$\bS_{i,j}=1$}
\STATE $\bL_{i,j} = (\bA_{i,j} - \sum_{k=1}^{j-1}\bL_{i,k}\bL_{j,k})/\bL_{j,j}$ \\
\ENDIF
\ENDFOR 
\STATE $\bL_{i,i} = (\bA_{i,i} - \sum_{k=1}^{i-1}\bL_{k,k})^{1/2}$
\ENDFOR
\end{algorithmic}
\label{alg:ic0}
\end{algorithm}

The incomplete Cholesky factorization \citep[e.g.,][]{Golub2012}, denoted by $\ichol(\bA, \bS)$ and given in Algorithm \ref{alg:ic0}, is identical to the standard Cholesky factorization of the matrix $\bA$, except that we skip all operations that involve elements that are not in the sparsity pattern represented by the zero-one matrix $\bS$. It is important to note that to compute $\bL=\ichol(\bA, \bS)$ for a large dense matrix $\bA$, we do not actually need to form or access the entire $\bA$; instead, to reduce memory usage and computational cost, we simply compute $\bL=\ichol(\bA \circ \bS, \bS)$ based on the sparse matrix $\bA \circ \bS$, where $\circ$ denotes element-wise multiplication. Thus, while we write expressions like $\bL=\ichol(\bA, \bS)$ for notational simplicity below, this should always be read as $\bL=\ichol(\bA \circ \bS, \bS)$.

For our HV approximation in \eqref{eq:vecchia-approx}, we henceforth set $\bS$ to be a sparse lower-triangular matrix with $\bS_{i,j}=1$ if $x_j \in \condset_i$ or if $i=j$, and $0$ otherwise. Thus, the sparsity pattern of $\bS$ is the same as that of $\bL$, and its transpose is that of $\bU$ shown in Figure \ref{fig:sparsity-pattern}. 

\begin{proposition} \label{prop:chol-lemma}
Assuming \eqref{eq:vecchia-approx}, denote $\var(\bx) = \hat\bfSigma$ and $\bL = \chol(\hat\bfSigma)$. Then, $\bL = \ichol(\bfSigma, \bS)$.
\end{proposition}

Hence, the Cholesky factor of the covariance matrix $\hat\bfSigma$ implied by the HV approximation can be computed using the incomplete Cholesky algorithm based on the (at most) $nN$ entries of the exact covariance $\bfSigma$ indicated by $\bS$. Using this result, we propose Algorithm \ref{alg:VF} for posterior inference on $\bx$ given $\by$. 

\begin{algorithm}[ht]
\caption{Posterior inference using HV: $\hv(\by,\bS,\bfmu,\bfSigma,\bH,\bR)$}
\KwInput{data $\by$; sparsity $\bS$; $\bfmu, \bfSigma$ s.t.\ $\bx \sim \normal_n(\bfmu, \bfSigma)$; obs.\ matrix $\bH$; noise variances $\bR$}
\KwResult{$\widetilde\bfmu$ and $\widetilde\bL$ such that $\hat{p}(\bx|\by) = \normal_n(\bx|\widetilde\bfmu,\widetilde\bL\widetilde\bL^\top)$ }
\begin{algorithmic}[1]
\STATE $\bL = \ichol(\bfSigma, \bS)$, using Algorithm \ref{alg:ic0}
\STATE $\bU = \bL^{-\top}$
\STATE $\bfLambda = \bU\bU^\top + \bH^\top\bR^{-1}\bH$
\STATE $\widetilde{\bU} = \bP\left(\chol(\bP\bfLambda\bP)\right)\bP$, where $\bP$ is the order-reversing permutation matrix
\STATE $\widetilde{\bL} = \widetilde{\bU}^{-\top}$
\STATE $\widetilde{\bfmu} = \bfmu + \widetilde{\bL}\widetilde{\bL}^\top\bH^{\top}\bR^{-1}\left(\by - \bH \bfmu\right)$
\end{algorithmic}
\label{alg:VF}
\end{algorithm}

By combining the incomplete Cholesky factorization with the results in Propositions \ref{prop:same-sparsity-factors} and \ref{prop:posterior} (saying that all involved Cholesky factors are sparse), we can perform fast posterior inference:
\begin{proposition} \label{prop:complexity-vf}
Algorithm \ref{alg:VF} can be carried out in $\order(n\condsetN^2)$ time and $\order(n\condsetN)$ space, assuming that $|\condset_i| \leq \condsetN$ for all $i=1,\ldots,n$.
\end{proposition}

%%%%%%%%%%%%%%%%%%
\subsection{Approximation accuracy}\label{sec:accuracy}

Formal error quantification for the HV approximation is a matter of ongoing research. Numerical simulations suggest that it can be almost as accurate (adjusting for its lower computational complexity) as the state-of-the-art general Vecchia approximation in many settings \citep[][Fig.~S3]{Katzfuss2018}. Generally, the quality of the approximation increases with the conditioning-set size $\condsetN$ and with the strength of the screening effect in the process to be approximated. For an approximately low-rank covariance matrix, only a small number of levels $M$ will be necessary to capture it accurately, while greater local variations might require a greater $M$. In our numerical experiments, we rely on the default strategy implemented in the \texttt{GPvecchia} package \citep{GPvecchia}, which we have found to produce good results. 
A drawback of the HV approximation is that its hierarchical conditional-independence assumptions can result in visual artifacts or edge effects along partition boundaries; some examples and a discussion of this issue are provided in Section S5. For situations in which these edge effects are of major concern, it may be useful to use a larger $N$ (and hence sacrifice some computational speed), to manually place the conditioning points at low levels near partition boundaries \citep[e.g.,][Sect.~2.5]{Katzfuss2015}, or to forego the HV approximation in favor of a different (Vecchia) approximation for purely spatial problems.
However, for spatio-temporal filtering, we are not aware of any other spatial approximation method which will ensure that the Cholesky factor of the filtering precision matrix will preserve its sparsity pattern when filtering across time points. This key advantage of our approach ensures scalability of the HV filter discussed in Section \ref{sec:filter} below. Using the HV approximation in a filtering setting has the added benefit that the boundary artifacts in the filtering mean field tend to fade or even disappear within a few time steps (see Section S5). This is because dependence between neighboring points is captured not only by the spatial covariance function, but also by the temporal evolution operator, which is particularly helpful in the case of diffusion or transport evolution models.

%%%%%%%%%%%%%%%%%%%%%%%%%%%%%%%%%%%%%%%%%%%%%%%%%%%%%%%
\section{Extensions to non-Gaussian spatial data using the Laplace approximation\label{sec:nongaussian}}

Now consider the model
\begin{align}
\label{eq:ngobs} y_{i} \,|\, \bx & \stackrel{ind}{\sim} g_i(y_{i} | x_{i}),\qquad i \in \obs, \\
\label{eq:ngevol} \bx  & \sim \normal_n(\bfmu,\bfSigma),
\end{align}
where $g_i$ is a distribution from an exponential family, and we slightly abuse notation by using $g_i(y_i|x_i)$ to denote both the distribution and its density. Using the HV approximation in \eqref{eq:vecchia-approx}--\eqref{eq:vecchia-approx-ind} for $\bx$, the implied posterior can be written as:
\begin{equation}
    \label{eq:nongaussianpost}
\hat{p}(\bx|\by) = \frac{p(\by|\bx) \hat{p}(\bx)}{\int p(\by|\bx) \hat{p}(\bx) d\bx} = \frac{(\prod_{i \in \obs} g_i(y_i|x_i)) \hat{p}(\bx)}{\int (\prod_{i \in \obs} g_i(y_i|x_i)) \hat{p}(\bx) d\bx}.
\end{equation}
Unlike in the Gaussian case as in \eqref{eq:gobs}, the integral in the denominator cannot generally be evaluated in closed form, and Markov Chain Monte Carlo methods are often used to numerically approximate the posterior. Instead, \citet{Zilber2019} proposed a much faster method that combines a general Vecchia approximation with the Laplace approximation \citep[e.g.][Sect.~3.4]{Tierney1986,Rasmussen2006}. The Laplace approximation is based on a Gaussian approximation of the posterior, obtained by carrying out a second-order Taylor expansion of the posterior log-density around its mode. Although the mode cannot generally be obtained in closed form, it can be computed straightforwardly using a Newton-Raphson procedure, because $\log \hat{p}(\bx|\by) = \log p(\by|\bx) + \log \hat{p}(\bx) + c$ is a sum of two concave functions and hence also concave (as a function of $\bx$, under appropriate parametrization of the $g_i$).

While each Newton-Raphson update requires the computation and decomposition of the $n \times n$ Hessian matrix, the update can be carried out quickly by making use of the sparsity implied by the Vecchia approximation. To do so, we follow \citet{Zilber2019} in exploiting the fact that the Newton-Raphson update is equivalent to computing the conditional mean of $\bx$ given pseudo-data. Specifically, at the $\ell$-th iteration of the algorithm, given the current state value $\bx^{(\ell)}$, let us define
\begin{equation}
 \textstyle \bu^{(\ell)} = \big[ u^{(\ell)}_i \big]_{i \in \obs}, \qquad \text{where} \quad u^{(\ell)}_i = \frac{\partial}{\partial x} \log g_i(y_i|x)\big\vert_{x=x^{(\ell)}_i},
\label{eq:u}
\end{equation}
and
\begin{equation}
\textstyle \bD^{(\ell)} = \diag\big(\{d_i^{(\ell)}: i \in \obs\}\big), \qquad \text{ where} \quad d_i^{(\ell)} = -\big(\frac{\partial^2}{\partial x^2} \log g_i(y_i|x)\big)^{-1}\big\vert_{x=x^{(\ell)}_i}.
\label{eq:D}
\end{equation}
Then, we compute the next iteration's state value $\bx^{(\ell+1)} = \mathbb{E}(\bx|\bt^{(\ell)})$ as the conditional mean of $\bx$ given pseudo-data
$\bt^{(\ell)} = \bx^{(\ell)} + \bD^{(\ell)} \bu^{(\ell)}$
assuming Gaussian noise,
$t^{(\ell)}_i | \bx \stackrel{ind.}{\sim} \normal(x_i,d_i^{(\ell)})$,
$i \in \obs$.
\citet{Zilber2019} recommend computing the conditional mean $\mathbb{E}(\bx|\bt^{(\ell)})$ based on a general-Vecchia-prediction approach proposed in \citet{Katzfuss2018}. Here, we instead compute the posterior mean using Algorithm \ref{alg:VF} based on the HV method described in Section \ref{sec:hv}, due to its sparsity-preserving properties. In contrast to the approach recommended in \citet{Zilber2019}, our algorithm is guaranteed to converge, because it is equivalent to Newton-Raphson optimization of the log of the posterior density in \eqref{eq:nongaussianpost}, which is concave.
Once the algorithm converges to the posterior mode $\widetilde\bfmu$, we obtain a Gaussian HV-Laplace approximation of the posterior as
\[
\hat{p}_L(\bx|\by) = \normal_n(\bx|\widetilde\bfmu,\widetilde\bL\widetilde\bL^\top),
\]
where $\widetilde\bL$ is the Cholesky factor of the negative Hessian of the log-posterior at $\widetilde\bfmu$.
Our approach is described in Algorithm \ref{alg:VL}. The main computational expense for each iteration of the \texttt{for} loop is carrying out Algorithm \ref{alg:VF}, and so each iteration requires only $\order(n\condsetN^2)$ time.

\begin{algorithm}[ht]
\caption{HV-Laplace inference: $\hvl(\by,\bS,\bfmu,\bfSigma,\{g_i\})$}
\KwInput{ data $\by$; sparsity $\bS$; $\bfmu, \bfSigma$ such that $\bx \sim \normal_n(\bfmu, \bfSigma)$; likelihoods $\{g_i: i \in \obs \}$}
\KwResult{ $\widetilde\bfmu$ and $\widetilde\bL$ such that $\hat{p}_L(\bx|\by) = \normal_n(\bx|\widetilde\bfmu,\widetilde\bL\widetilde\bL^\top)$ }
\begin{algorithmic}[1]
\STATE Initialize $\bx^{(0)} = \bfmu$
\STATE Set $\bH = \bI_{\obs,:}$ as the rows $\obs$ of the $n \times n$ identity matrix $\bI$
\FOR{$\ell = 0,1,2,\ldots$}
    \STATE Calculate $\bu^{(\ell)}$ as in \eqref{eq:u}, $\bD^{(\ell)}$ as in \eqref{eq:D}, and $\bt^{(\ell)} = \bx^{(\ell)} + \bD^{(\ell)} \bu^{(\ell)}$
    \STATE Calculate $[\bx^{(\ell+1)},\widetilde\bL]=\hv(\bt^{(\ell)},\bS,\bfmu,\bfSigma,\bH,\bD^{(\ell)})$ using Algorithm \ref{alg:VF}
    \IF{$\norm{\bx^{(\ell+1)} - \bx^{(\ell)}}/\norm{\bx^{(\ell)}} < \epsilon$}
    \BREAK
    \ENDIF
\ENDFOR
\RETURN $\widetilde{\bfmu}=\bx^{(\ell+1)}$ and $\widetilde\bL$
\end{algorithmic}
\label{alg:VL}
\end{algorithm}

In the Gaussian case, when $g_i(y_i | x_i) = \normal(y_i|a_i x_i,\tau_i^2)$ for some $a_i \in \mathbb{R}$, it can be shown using straightforward calculations that the pseudo-data $t_i = y_i/a_i$ and pseudo-variances $d_i = \tau_i^2$ do not depend on $\bx$, and so Algorithm \ref{alg:VL} converges in a single iteration. If, in addition, $a_i=1$ for all $i=1,\ldots,n$, then \eqref{eq:ngobs} becomes equivalent to \eqref{eq:gobs}, and Algorithm \ref{alg:VL} simplifies to Algorithm \ref{alg:VF}.

For non-Gaussian data, our HVL algorithm shares the limitations of the Laplace approximation. In particular, it approximates the mean of the filtering distribution with its mode, which can result in significant errors for highly non-Gaussian settings. Furthermore, because the variance estimate is based on the curvature of the likelihood at the mode, it is prone to underestimating the true uncertainty. \citet{Nickisch2008} provide further details and a comprehensive comparison of several methods of approximating non-Gaussian likelihoods.
% In addition, Laplace inference is based on the exact Gaussian likelihood of the latent process. Our method allows scalability by approximating this likelihood which introduces additional error, as indicated in Section \ref{sec:accuracy}. 
Notwithstanding these known limitations, empirical studies \citep[e.g.,][]{Bonat2016, Zilber2019} have shown that Laplace-type approximations can be effectively applied to environmental data sets and can strongly outperform sampling-based approaches such as Markov chain Monte Carlo.

%%%%%%%%%%%%%%%%%%%%%%%%%%%%%%%%%%%%%%%%%%%%%%%%%%%%%%%
\section{Fast filters for spatio-temporal models} \label{sec:filter}

%%%%%%%%%%%%%%%%%%%%%%%
\subsection{Linear evolution} \label{sec:linear-filter}
We now turn to a spatio-temporal state-space model (SSM), which adds a temporal evolution model to the spatial model \eqref{eq:ngobs} considered in Section \ref{sec:nongaussian}. For now, assume that the evolution is linear.
Starting with an initial distribution $\bx_0 \sim \normal_{n}(\bfmu_{0|0},\bfSigma_{0|0})$, we consider the following SSM for discrete time $t=1,2,\ldots$:
\begin{align}
y_{ti} \,|\, \bx_t & \stackrel{ind}{\sim} g_{ti}(y_{ti} | x_{ti}), \qquad i \in \obs_t \label{eq:linear-dynamics-obs}\\
\bx_t \,|\, \bx_{t-1} & \sim \normal_n(\levol_t\bx_{t-1},\bQ_t),
\label{eq:linear-dynamics}
\end{align}
where $\by_{t}$ is the data vector consisting of $n_t \leq n$ observations $\{y_{ti}: i \in \obs_t\}$, $\obs_t \subset \{1,\ldots,n\}$ contains the observation indices at time $t$, $g_{ti}$ is a distribution from the exponential family, $\bx_t = (x_1,\ldots,x_n)^\top$ is the latent spatial field of interest at time $t$ observed at a spatial grid $\grid$, and $\levol_t$ is a sparse $n \times n$ evolution matrix. Note that we allow for different observation locations at each time point, which is helpful in many remote-sensing applications where observations are often missing (e.g., due to changing cloud cover).
The methods we propose are suitable for any type of $\bQ_t$ matrices; due to its hierarchical structure, HV is able to successfully capture both long-range dependence (which corresponds to low-rank structures in $\bQ_t$) as well as finer, local variations. 

At time $t$, our goal is to obtain or approximate the filtering distribution $p(\bx_t|\by_{1:t})$ of $\bx_t$ given data $\by_{1:t}$ up to the current time $t$. This task, also referred to as data assimilation or on-line inference, is commonly encountered in many fields of science whenever one is interested in quantifying the uncertainty in the current state or in obtaining forecasts into the future.
If the observation equations $g_{ti}$ are all Gaussian, the filtering distribution can be derived using the Kalman filter \citep{Kalman1960} for small to moderate $n$. At each time $t$, the Kalman filter consist of a forecast step that computes $p(\bx_t|\by_{1:t-1})$, and an update step which then obtains $p(\bx_t|\by_{1:t})$. For linear Gaussian SSMs, both of these distributions are multivariate normal.

Our Kalman-Vecchia-Laplace (KVL) filter extends the Kalman filter to high-dimensional SSMs (i.e., large $n$) with non-Gaussian data, as in \eqref{eq:linear-dynamics-obs}--\eqref{eq:linear-dynamics}. % assimilate non-Gaussian observations
Its update step is very similar to the inference problem in Section \ref{sec:nongaussian}, and hence it essentially consists of the HVL in Algorithm \ref{alg:VL}. 
We complement this update with a forecast step, in which the moment estimates are propagated forward using the temporal evolution model. This forecast step is exact, and so the KVL approximation error is solely due to the HVL approximation at each update step. 

\begin{algorithm}[ht]
\caption{Kalman-Vecchia-Laplace (KVL) filter}
\KwInput{$\bS$, $\bfmu_{0|0}$, $\bfSigma_{0|0}$, $\{(\by_t,\levol_t,\bQ_t,\{g_{t,i}\}): t=1,2,\ldots \}$}
\KwResult{$\bfmu_{t|t}, \bL_{t|t}$, such that $\hat{p}(\bx_t|\by_{1:t}) = \normal_n(\bx_t|\bfmu_{t|t},\bL_{t|t}\bL_{t|t}^\top)$}
\begin{algorithmic}[1]
\STATE Compute $\bU_{0|0} = \ichol(\bfSigma_{0|0}, \bS)$ and $\bL_{0|0} = \bU_{0|0}^{-\top}$
\FOR{$t = 1, 2, \dots$}
    % \STATE Acquire $\by_t$
    \STATE Forecast: $\bfmu_{t|t-1} = \levol_t\bfmu_{t-1|t-1}$ and $\bL_{t|t-1} = \levol_t\bL_{t-1|t-1}$ \label{ln:evolL}
    \STATE For all $(i,j)$ with $\bS_{i,j}=1$: $\bfSigma_{t|t-1; i,j} = \bL_{t|t-1;i,:} \bL_{t|t-1;j,:}^\top + \bQ_{t;i,j}$ \label{ln:forecov}
    \STATE Update: $[\bfmu_t, \bL_{t|t}] = \hvl(\by_t, \bS, \bfmu_{t|t-1}, \bfSigma_{t|t-1}, \left\{g_{t,i}\right\}$) using Algorithm \ref{alg:VL}
    \RETURN $\bfmu_{t|t}, \bL_{t|t}$
\ENDFOR
\end{algorithmic}
\label{alg:KVLfilter}
\end{algorithm}

The KVL filter is given in Algorithm \ref{alg:KVLfilter}.
In Line \ref{ln:forecov}, $\bL_{t|t-1;i,:}$ denotes the $i$th row of $\bL_{t|t-1}$. The KVL filter scales well with the state dimension $n$. The evolution matrix $\levol_t$, which is often derived using a forward-finite-difference scheme and thus has only a few nonzero elements in each row, can be quickly multiplied with $\bL_{t-1|t-1}$ in Line \ref{ln:evolL}, as the latter is sparse (see Section \ref{sec:vecchia-complexity}). The $\order(n\condsetN)$ necessary entries of $\bfSigma_{t|t-1}$ in Line \ref{ln:forecov} can also be calculated quickly due to the sparsity of $\bL_{t|t-1;i,:}$. The low computational cost of the \texttt{HVL} algorithm has already been discussed in Section \ref{sec:nongaussian}. Thus, assuming sufficiently sparse $\levol_t$, the KVL filter scales approximately as $\order(n\condsetN^2)$ per iteration.
In the case of Gaussian data (i.e., all $g_{ti}$ in \eqref{eq:linear-dynamics-obs} are Gaussian), our KVL filter will produce similar filtering distributions as the more complicated multi-resolution filter of \citet{Jurek2018}.

%%%%%%%%%%%%%%%%%%%%%%%
\subsection{An extended filter for nonlinear evolution\label{sec:non-linear-filter}}

Finally, we consider a nonlinear and non-Gaussian model, which extends \eqref{eq:linear-dynamics-obs}--\eqref{eq:linear-dynamics} by allowing nonlinear evolution operators, $\evol_t: \mathbb{R}^n \rightarrow \mathbb{R}^n$. This results in the model
\begin{align}
\label{eq:obs} y_{ti} \,|\, \bx_t & \stackrel{ind}{\sim} g_{ti}(y_{ti} | x_{ti}), \qquad i \in \obs_t \\
\label{eq:evol} \bx_t \,|\, \bx_{t-1} & \sim \normal_n(\evol_t(\bx_{t-1}),\bQ_t).
\end{align}

Due to the nonlinearity of the evolution operator $\evol_t$, the KVL filter in Algorithm \ref{alg:KVLfilter} is not directly applicable anymore. However, similar inference is still possible as long as the evolution is not too far from linear. Approximating the evolution as linear is generally reasonable if the time steps are short, or if the measurements are highly informative. In this case, we propose the extended Kalman-Vecchia-Laplace filter (EKVL) in Algorithm~\ref{alg:eKVLfilter}, which approximates the extended Kalman filter \citep[e.g.,][Ch.~5]{Grewal1993} and extends it to non-Gaussian data using the Vecchia-Laplace aproach. For the forecast step, EKVL computes the forecast mean as $\bfmu_{t|t-1} = \evol_t(\bfmu_{t-1|t-1})$. The forecast covariance matrix $\bfSigma_{t|t-1}$ is obtained as before, after approximating the evolution using the Jacobian as $\levol_t = \frac{\partial \evol_t(\by_{t-1})}{\partial \by_{t-1}} \big|_{\by_{t-1} = \bfmu_{t-1|t-1} }$.  Errors in the forecast covariance matrix due to this linear approximation can be captured in the model-error covariance, $\bQ_t$. If the Jacobian matrix cannot be computed, it is sometimes possible to build a statistical emulator \citep[e.g.,][]{Kaufman2011} instead, which approximates the true evolution operator.

Once $\bfmu_{t|t-1}$ and $\bfSigma_{t|t-1}$ have been obtained, the update step of the EKVL proceeds exactly as in the KVL filter by approximating the forecast distribution as Gaussian.

\begin{algorithm}[ht]
\caption{Extended Kalman-Vecchia-Laplace (EKVL) filter}
\KwInput{$\bS$, $\bfmu_{0|0}$, $\bfSigma_{0|0}$, $\{(\by_t,\evol_t,\bQ_t,\{g_{t,i}\}): t=1,2,\ldots \}$}
\KwResult{$\bfmu_{t|t}, \bL_{t|t}$, such that $\hat{p}(\bx_t|\by_{1:t}) = \normal_n(\bx_t|\bfmu_{t|t},\bL_{t|t}\bL_{t|t}^\top)$}
\begin{algorithmic}[1]
\STATE Compute $\bU_{0|0} = \ichol(\bfSigma_{0|0}, \bS)$ and $\bL_{0|0} = \bU_{0|0}^{-\top}$
\FOR{$t = 1, 2, \dots$}
    \STATE Calculate $\levol_t = \frac{\partial \evol_t(\bx_{t-1})}{\partial \bx_{t-1}} \big|_{\bx_{t-1} = \bfmu_{t-1|t-1} }$ \label{ln:jacobian}
    \STATE Forecast: $\bfmu_{t|t-1} = \evol_t(\bfmu_{t-1|t-1})$ and $\bL_{t|t-1} = \levol_t\bL_{t-1|t-1}$
     \STATE For all $(i,j)$ with $\bS_{i,j}=1$: $\bfSigma_{t|t-1; i,j} = \bL_{t|t-1;i,:} \bL_{t|t-1;j,:}^\top + \bQ_{t;i,j}$
    \STATE Update: $[\bfmu_t, \bL_{t|t}] = \hvl(\by_t, \bS, \bfmu_{t|t-1}, \bfSigma_{t|t-1}, \left\{g_{t,i}\right\}$) using Algorithm \ref{alg:VL}
    \RETURN $\bfmu_{t|t}, \bL_{t|t}$
\ENDFOR
\end{algorithmic}
\label{alg:eKVLfilter}
\end{algorithm}

Similarly to Algorithm \ref{alg:KVLfilter}, EKVL scales very well with the dimension of $\bx$, the only difference being the additional operation of calculating the Jacobian in Line \ref{ln:jacobian}, whose cost is problem dependent. Only those entries of $\bE_t$ need to be calculated that are multiplied with non-zero entries of $\bL_{t-1|t-1}$, whose sparsity structure is known ahead of time.

%%%%%%%%%%%%%%%%%%%%%%%%%%%%%%%%%%%%%%%%%%%%%%%%%%%%%%%
\section{Numerical comparison\label{sec:comparison}}

\subsection{Methods and criteria\label{sec:simul-general}}

We considered and compared the following methods:
\begin{description}
\item[\textbf{Hierarchical Vecchia (HV)}:] Our methods as described in this paper.
\item[\textbf{Low rank (LR)}:] A special case of HV with $M=1$, in which the diagonal and the first $N$ columns of $\bS$ are nonzero, and all other entries are zero. This results in a matrix approximation $\hat\bfSigma$ that is of rank $\condsetN$ plus diagonal, known as the modified predictive process \citep[][]{Banerjee2008,Finley2009} in spatial statistics. LR has the same computational complexity as HV.
\item[\textbf{Dense Laplace (DL)}:] A further special case of HV with $M=0$, in which $\bS$ is a fully dense matrix of ones. Thus, there is no error due to the Vecchia approximation, and so in the non-Gaussian spatial-only setting, this is equivalent to a Gaussian Laplace approximation. DL will generally be more accurate than HV and low-rank, but it scales as $\order(n^3)$ and is thus not feasible for high dimension $n$.
\end{description}

For each scenario below, we simulated observations using \eqref{eq:obs}, taking $g_{t,i}$ to be each of four exponential-family distributions: Gaussian, $\normal(x,\tau^2)$; logistic Bernoulli, $\mathcal{B}(1/(1+e^{-x}))$; Poisson, $\mathcal{P}(e^x)$; and gamma, $\mathcal{G}(a, ae^{-x})$, with shape parameter $a=2$. For most scenarios, we assumed a moderate state dimension $n$, so that DL remained feasible; a large $n$ was considered in Section \ref{sec:big-simul}. In the case of non-Gaussian observations, we used $\epsilon = 10^{-5}$ in Algorithm \ref{alg:VL}.

The main metric to compare HV and LR was the difference in KL divergence between their posterior or filtering distributions and those generated by DL; as the exact distributions were not known here, we approximated the divergence by the difference in joint log scores \citep[dLS; e.g.,][]{Gneiting2014} calculated at each time point for the joint posterior or filtering distribution of the entire field and averaged over several simulations. We also calculated the relative root mean square prediction error (RRMSPE), defined as the root mean square prediction error of HV and LR, respectively, divided by the root mean square prediction error of DL; the error is calculated with respect to the true simulated latent field. For both criteria, lower values are better. In each simulation scenario, the scores were averaged over a sufficiently large number of repetitions to achieve stable averages.

For comparisons over a range of conditioning-set sizes $\condsetN$, we supplied a set of desired sizes to the \texttt{GPvecchia} package, upon which our method implementation is based. For a desired size, \texttt{GPvecchia} automatically determines a suitable partitioning scheme (and hence a suitable $N$) for the grid $\locs$, such that $\condsetN$ is equal to or slightly below the desired value. We generated the HV approximation using this approach and then used the resulting $\condsetN$ value for both LR and DL, to ensure a fair comparison.

Large-scale spatio-temporal filtering is also often achieved using the ensemble Kalman filter (EnKF) and its extensions. While EnKF methods can be very effective, even in the presence of nonlinear evolution, they often struggle with long-range correlations in the forecast covariance, because EnKF methods typically require localization (e.g., tapering the empirical covariance matrix). As the EnKF is a substantially different approach that does not neatly fit into our framework here, we conducted a separate comparison between HV and EnKF in Section S3 in the Supplementary Materials; the results strongly favored the HV filter.

%%%%%%%%%%%%%
\subsection{Spatial-only data}
\label{sec:purely-spatial}

In our first scenario, we considered spatial-only data according to \eqref{eq:ngobs}--\eqref{eq:ngevol} on a grid $\grid$ of size $n=34 \times 34 = 1{,}156$ on the unit square, $\domain = [0,1]^2$. We set $\bfmu = \bfzero$ and $\bfSigma_{i,j} = \exp(-\|\bs_i - \bs_j\|/0.15)$. For the Gaussian likelihood, we assumed variance $\tau^2 = 0.2$. %Hierarchical Vecchia parameters for each N are presented in Table \ref{tab:vecchiaParams}.

%\begin{table}[]
%    \centering
%    \begin{tabular}{c|c|c|c}
%         N & \# resolutions ($M$) & \lVert $\sx_\jm$  \lVert& \# partitions ($J$) \\
%         \hline
%         1 & 1 & 1 & 2500 \\
%         2 & 2 & 1, 1 & 2\\
%         4 & 4 & 1, 1, 1, 1 & 2\\
%         9 & 9 & 1, 1, 1, 1, 1, 1, 1, 1, 1 & 2\\
%         16 & 11 & 1, 1, 1, 1, 2, 2, 2, 2, 2, 2, $\sim$ & 2\\
%         32 & 10 & 3, 3, 3, 3, 4, 4, 4, 4, 4, $\sim$ & 2\\
%         51 & 9 & 6, 6, 6, 6, 6, 6, 6, 6, $\sim$ & 2\\
%         71 & 8 & 8, 9, 9, 9, 9, 9, 9, 9 & 2\\
%         85 & 8 & 11, 11, 11, 11, 11, 11, 11, $\sim$ & 2\\
%    \end{tabular}
%    \caption{Parameters used in the Vecchia approximation. The symbol $\sim$ indicates that not all sets $\sx_\jM$ had the same number of elements. This is because keeping a constant number of elements was impossible as there were not enough "unused" grid points left.}
%    \label{tab:vecchiaParams}
%\end{table}

The comparison scores averaged over 100 simulations for the posteriors obtained using Algorithm \ref{alg:VL} are shown as a function of $\condsetN$ in Figure \ref{fig:spatial-scores}. HV (Algorithm \ref{alg:VF}) was much more accurate than LR for each value of $\condsetN$.

\begin{figure}[ht]
    \centering
    \begin{subfigure}{1.0\textwidth}
        \includegraphics[trim=0mm 6mm 0mm 0mm, clip,width=1.0\textwidth]{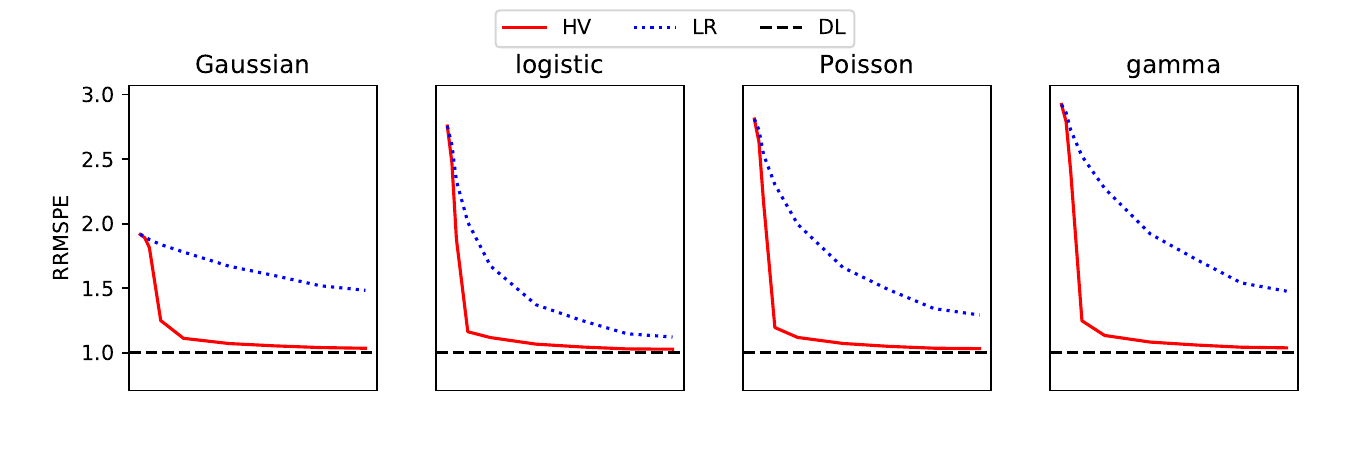}
        % \caption{RRMSPE (relative to pure Laplace approximation)}
    \end{subfigure}
    \medskip
    \begin{subfigure}{1.0\textwidth}
        \includegraphics[trim=0mm 0mm 0mm 6mm, clip,width=1.0\textwidth]{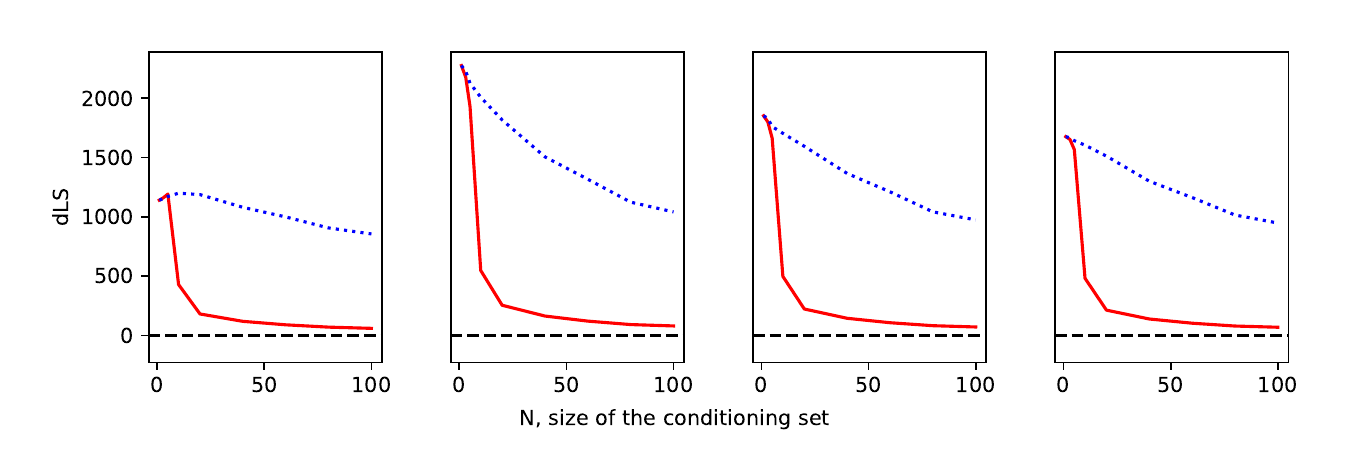}
        % \caption{difference in log score}
    \end{subfigure}
    \caption{Approximation accuracy for the posterior distribution $\bx|\by$ for spatial data (see Section \ref{sec:purely-spatial})}
    \label{fig:spatial-scores}
\end{figure}

%%%%%%%%%%%%%
\subsection{Linear temporal evolution\label{sec:simul-linear}}

Next, similar to \citet{Jurek2018}, we considered a linear spatio-temporal advection-diffusion process $x(\bs, t)$ whose evolution is governed by the partial differential equation,
\[
\frac{\partial x}{\partial t} = \alpha \left( \frac{ \partial^2 x }{ \partial^2 s_x } + \frac{ \partial^2 x }{ \partial^2 s_y } \right) + \beta \left( \frac{ \partial x}{\partial s_x} + \frac{\partial x}{\partial s_y} \right) + \nu(\bs, t),
\]
where $\bs = (s_x, s_y) \in  \domain=[0,1]^2$ and $t \in [0, T]$. We set the diffusion parameter $\alpha = 4 \times 10^{-5}$ and advection parameter $\beta = 10^{-2}$. The error $\nu(\bs, t)$ was a zero-mean stationary Gaussian process with exponential covariance function with range $\lambda = 0.15$, independent across time. 
% The resulting evolution model in matrix form \eqref{eq:linear-dynamics} can be found in \citet{Jurek2018}.

The spatial domain $\domain$ was discretized on a grid of size $n=34 \times 34 = 1{,}156$ using centered finite differences, and we considered discrete time points $t=1,\ldots,T$ with $T=20$. After this discretization, our model was of the form \eqref{eq:linear-dynamics-obs}--\eqref{eq:linear-dynamics}, where $\bfSigma_{0|0}=\bQ_1=\ldots=\bQ_T$ with $(i,j)$th entry $\exp(-\|\bs_i - \bs_j\|/\lambda)$, and $\bE_t$ was a sparse matrix with nonzero entries corresponding to interactions between neighboring grid points to the right, left, top and bottom. See the supplementary material of \citet{Jurek2018} for details.

At each time $t$, we generated $n_t = n/10$ observations with indices $\obs_t$ sampled randomly from $\{1,\ldots,n\}$. For the Gaussian case, we assumed variance $\tau^2=0.25$.
We used conditioning sets of size at most $\condsetN=41$ for both HV and LR; specifically, for HV, we used $J=2$ partitions at $M=7$ levels, with set sizes $|\sx_\jm|$ of $5, 5, 5, 5, 6, 6, 6$, respectively, for $m=0,1,\ldots,M-1$, and $|\sx_\jM| \leq 3$.

Figure \ref{fig:linear-scores} compares the scores for the filtering distributions $\bx_t | \by_{1:t}$ obtained using Algorithm \ref{alg:KVLfilter}, averaged over 80 simulations. Again, HV was much more accurate than LR. Importantly, while the accuracy of HV was relatively stable over time, LR became less accurate over time, with the approximation error accumulating.

\begin{figure}[ht]
    \centering
    \begin{subfigure}{1.0\textwidth}
        \includegraphics[trim=0mm 6mm 0mm 0mm, clip,width=1.0\textwidth]{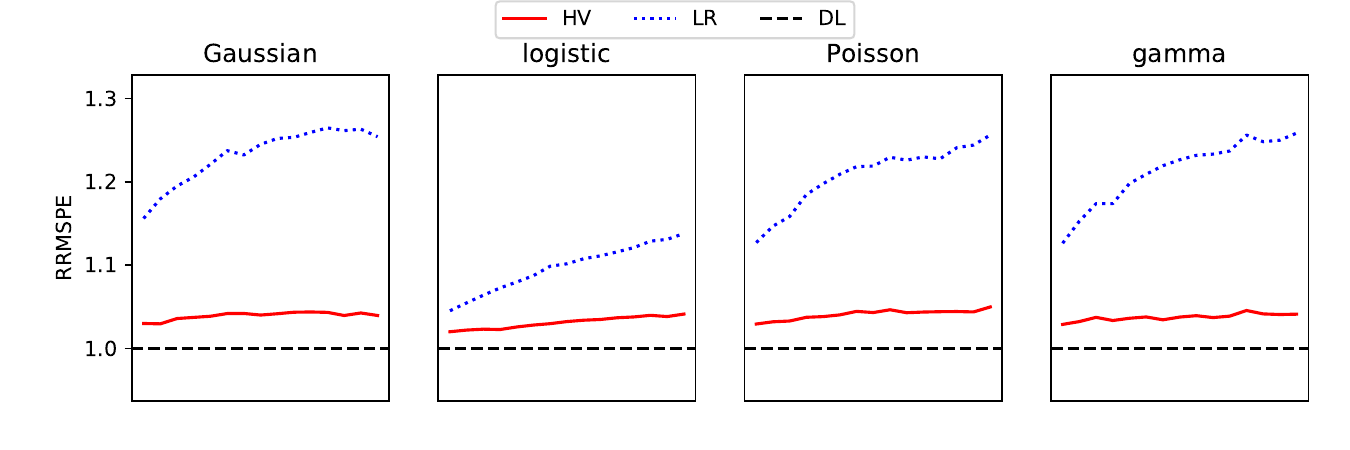}
        % \caption{RRMSPE (relative to pure Laplace approximation)}
    \end{subfigure}
    \begin{subfigure}{1.0\textwidth}
        \includegraphics[trim=0mm 0mm 0mm 6mm, clip,width=1.0\textwidth]{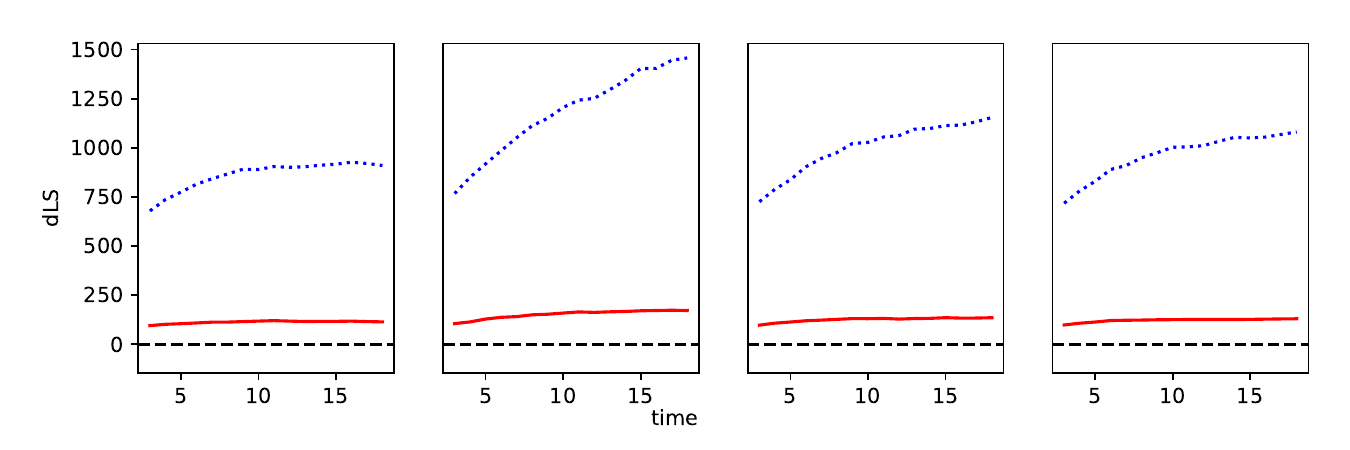}
        % \caption{difference in log score}
    \end{subfigure}
    \caption{Accuracy of filtering distributions $\bx_t | \by_{1:t}$ for the advection-diffusion model in Section \ref{sec:simul-linear}}
    \label{fig:linear-scores}
\end{figure}

%%%%%%%%%%%%%
\subsection{Simulations using a very large \texorpdfstring{$n$}{n}} \label{sec:big-simul}

We repeated the advection-diffusion experiment from Section \ref{sec:simul-linear} on a high-resolution grid of size $n=300 \times 300 = 90{,}000$, with $n_t=9{,}000$ observations corresponding to 10\% of the grid points. In order to avoid numerical artifacts related to the finite differencing scheme, we reduced the advection and diffusion coefficients to $\alpha = 10^{-7}$ and $\beta=10^{-3}$, respectively. We set $\condsetN=44$, $M=14$, $J=2$, and $|\sx_\jm| = 3$ for $m=0,1,\ldots,M-1$, and $|\sx_\jM| \leq 2$. DL was too computationally expensive due to the high dimension $n$, and so we simply compared HV and LR based on the root mean square prediction error (RMSPE) between the true state and their respective filtering means, averaged over 10 simulations.

As shown in Figure \ref{fig:linear-scores-big}, HV was again much more accurate than LR. Comparing to Figure \ref{fig:linear-scores}, we see that the relative improvement of HV to LR increased even further; taking the Gaussian case as an example, the ratio of the RMSPE for HV and LR was around 1.2 in the small-$n$ setting, and greater than 2 in the large-$n$ setting.

\begin{figure}[ht]
    \centering
    \begin{subfigure}{1.0\textwidth}
        \includegraphics[width=1.0\textwidth]{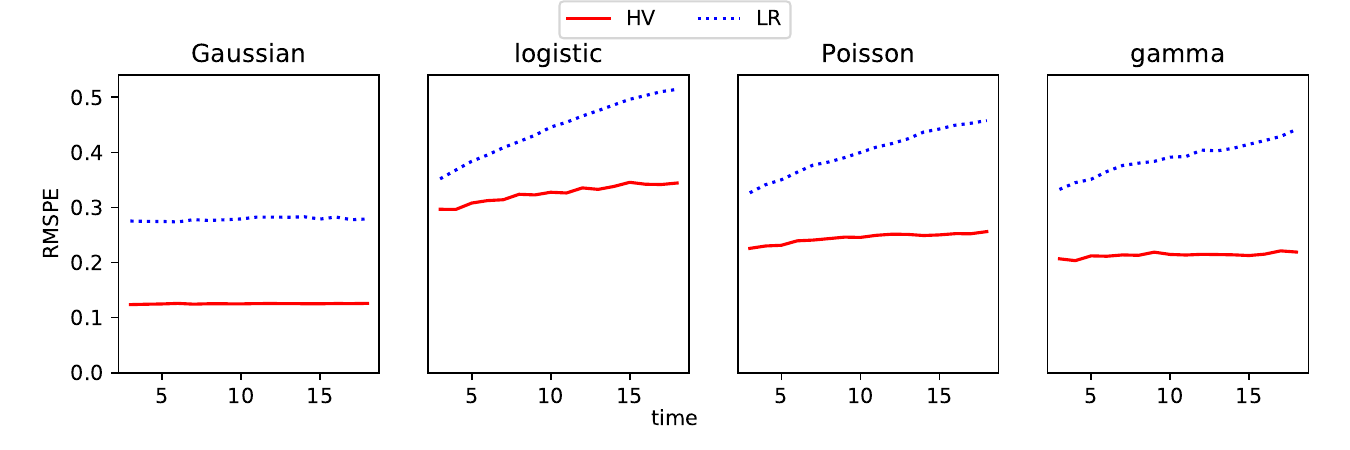}
    \end{subfigure}
    \caption{Root mean square prediction error (RMSPE) for the filtering mean in the high-dimensional advection-diffusion model with $n=90{,}000$ in Section \ref{sec:big-simul}}
    \label{fig:linear-scores-big}
\end{figure}

%%%%%%%%%%%%%
\subsection{Nonlinear evolution with non-Gaussian data\label{sec:lorenz}}

%Simple model for testing: https://arxiv.org/abs/1708.08794 %Simulations: see Model III in http://journals.ametsoc.org/doi/abs/10.1175/jas3430.1 

Our final set of our simulations involves the most general model \eqref{eq:obs}-\eqref{eq:evol} with nonlinear evolution $\evol_t$. Specifically, we considered a complicated model \citep[][Sect.~3]{Lorenz2005} that realistically replicates many features of atmospheric variables along a latitudinal band, and its different variants are commonly used as a test-bed for other data assimilation techniques \citep[e.g.,][]{Ott2004}. %A special case of this model \citep{Lorenz1996} is an important benchmark for data-assimilation techniques. \citep{Houtekamer2016, Bannister2017,Carrassi2018}
The model dynamics are described by
\begin{equation}
\textstyle    \frac{\partial }{\partial t} \tilde x_i = \frac{1}{K^2} \sum_{l=-K/2}^{K/2}\sum_{j=-K/2}^{K/2} -\tilde x_{i-2K-l}\tilde x_{i-K-j} + \tilde x_{i-K+j-l}\tilde x_{i+K+j} - \tilde x_i + F,
    \label{eq:Lorenz04-1}
\end{equation}
where $\tilde x_{-i} = \tilde x_{n-i}$, and we used $K=32$ and $F=10$.
By solving \eqref{eq:Lorenz04-1} on a regular grid of size $n=960$ on a circle with unit circumference using a 4-th order Runge-Kutta scheme, with five internal steps of size $dt=0.005$, and setting $x_i = b\tilde x_i$ with $b=0.2$, we obtained the evolution operator $\evol_t$. We also calculated an analytic expression for its derivative $\nabla\evol_t$, which is necessary for Algorithm \ref{alg:eKVLfilter}.

To complete our state-space model \eqref{eq:obs}-\eqref{eq:evol}, we assumed $\bQ_{t;i,j} = 0.2 \exp(-\|\bs_i - \bs_j\|/0.15)$, we randomly selected $n_t = n/10 = 96$ observation indices at each $t$, and we took the initial moments $\bfmu_{0|0}$ and $\bfSigma_{0|0}$ to be the corresponding sample moments from a long simulation from \eqref{eq:Lorenz04-1}. We simulated 40 datasets from the state-space model, each at $T=20$ time steps, and for each of the four exponential-family likelihoods, using $\tau^2=0.2$ in the Gaussian case.

%A sample realization of the process at selected timepoints is shown in Figure \ref{fig:lorenz}
Section S4 contains sample trajectories from this Lorenz model and shows that the trajectories can be successfully recovered using the EKVL filter (Algorithm \ref{alg:eKVLfilter}), even based on non-Gaussian data.

For each simulated data set, we applied the three filtering methods described in Section \ref{sec:simul-general}. We used $\condsetN=39$, and for the EKVL filter we set $J=2$, $M=7$, and $\lvert \sx_\jm \rvert$ equal to 5, 5, 5, 5, 6, 6, 6, respectively, for $m=0, 1, \dots, M-1$, and $\lvert \sx_\jM\rvert \leq 1$. The average scores over the 40 simulations are shown in Figure \ref{fig:lorenz-scores}.
Our method (HV) compared favorably to the low-rank filter and provided excellent approximation accuracy as evidenced by very low RRMSPE and dLS scores.

\begin{figure}[ht]
    \centering
    \begin{subfigure}{1.0\textwidth}
        \includegraphics[trim=0mm 6mm 0mm 0mm, clip,width=1.0\textwidth]{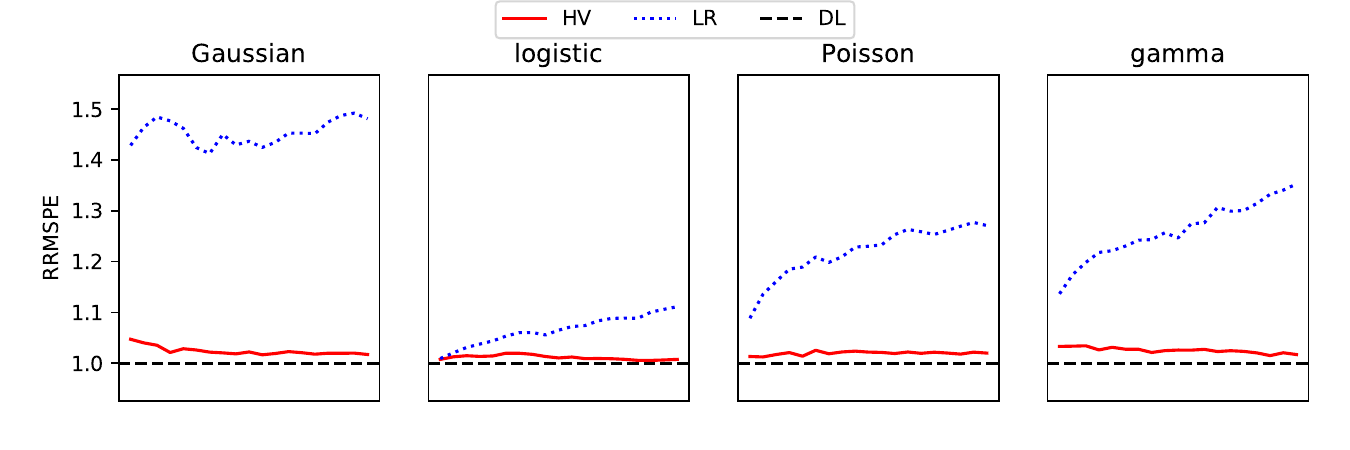}
        %\caption{RRMSPE (relative to pure Laplace approximation)}
    \end{subfigure}
    \begin{subfigure}{1.0\textwidth}
        \includegraphics[trim=0mm 0mm 0mm 12mm, clip,width=1.0\textwidth]{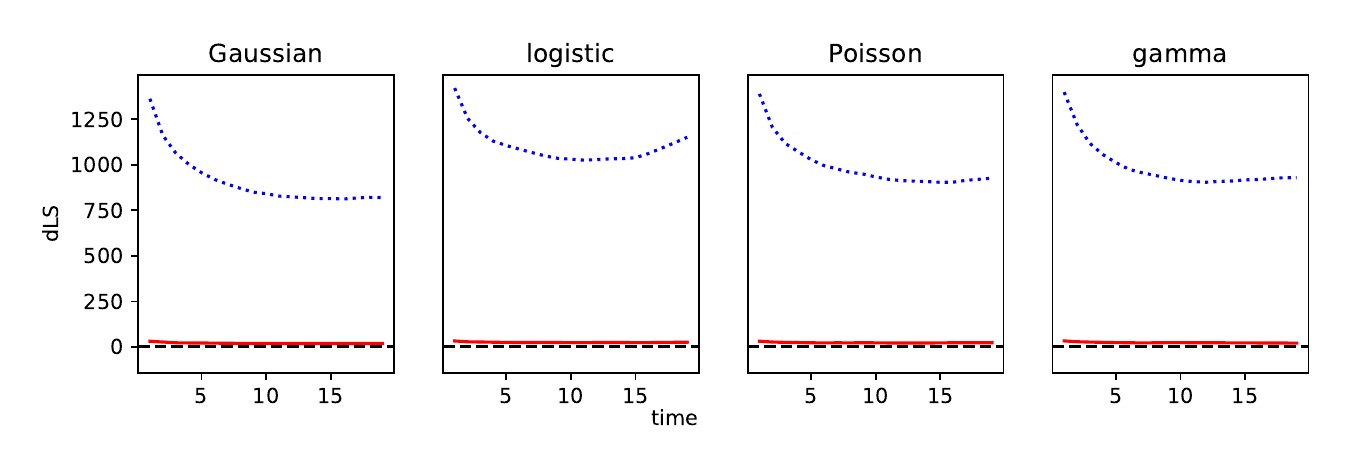}
        %\caption{difference in log score}
    \end{subfigure}
    \caption{Accuracy of filtering distribution $\bx_t | \by_{1:t}$ for the Lorenz model in Section \ref{sec:lorenz}}
    \label{fig:lorenz-scores}
\end{figure}

%%%%%%%%%%%%%%%%%%%%%%%%%%%%%%%%%%%%%%%%%%%%%%%%%%%%%%%
\section{Filtering analysis of satellite data\label{sec:data-application}}

We also applied the filtering method described in Algorithm \ref{alg:eKVLfilter} to measurements of total precipitable water vapor (TPW) acquired by the Microwave Integrated Retrieval System (MIRS) instrument mounted on NASA's Geostationary Operational Environmental Satellite (GOES). This dataset was previously analyzed in \citet{katzfuss2017parallel}. We used 11 consecutive sets of observations collected over a period of 40 hours in January 2011 on a grid of size $263 \times 263 = 69{,}169$ over the continental United States, with each cell covering a square of size $16 \times 16$km. In order to avoid specifying boundary conditions for the diffusion model below, we restricted the area of interest to a smaller grid of size $n = 247 \times 247 = 61{,}009$. Observations at selected time points are shown in Figure \ref{fig:TPW-filtering}.

\begin{figure}[ht]
    \centering
        \includegraphics[width=.7\textwidth]{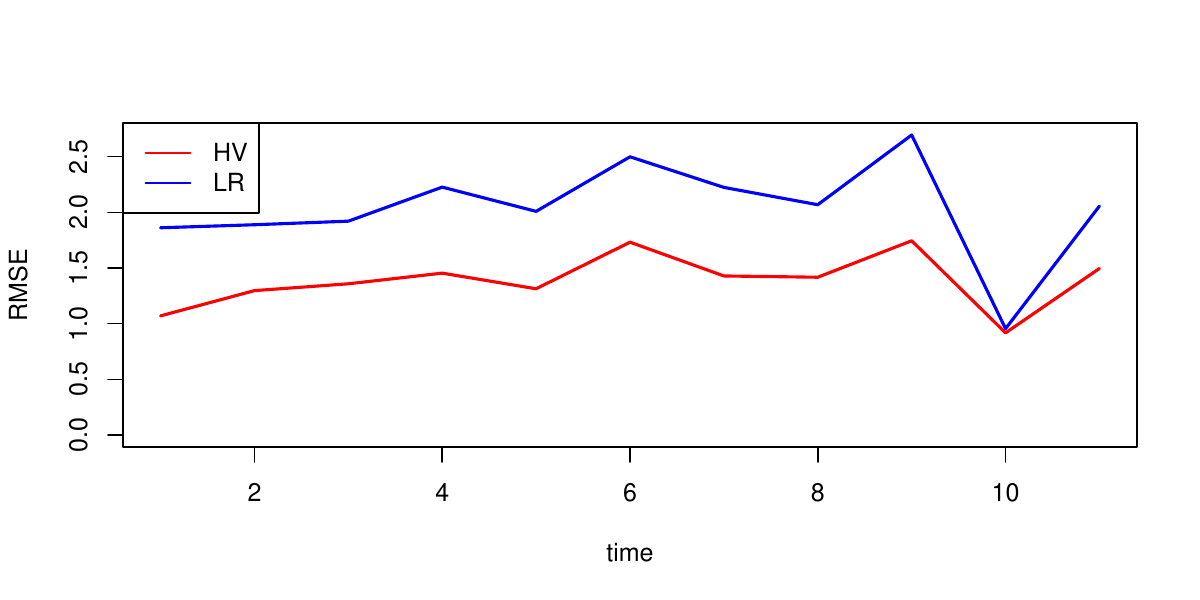}
    \caption{Root mean square prediction error (RMSPE) for the HV (red) and LR (blue) filters versus the 11 time points in the satellite-data application}
    \label{fig:data-application-scores}
\end{figure}

\begin{figure}[ht]
    \begin{subfigure}{0.325\textwidth}
        \includegraphics[width=1.0\textwidth]{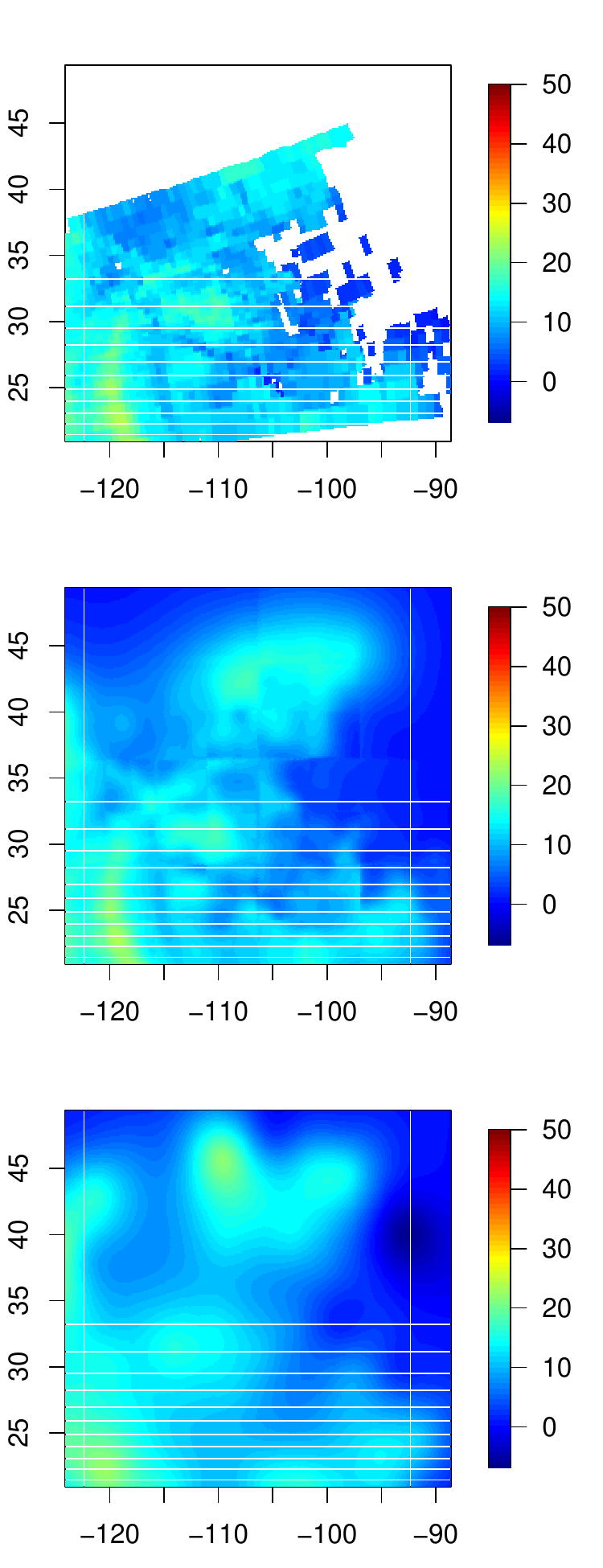}
        \caption*{t = 1}
    \end{subfigure}
    \begin{subfigure}{0.325\textwidth}
        \includegraphics[width=1.0\textwidth]{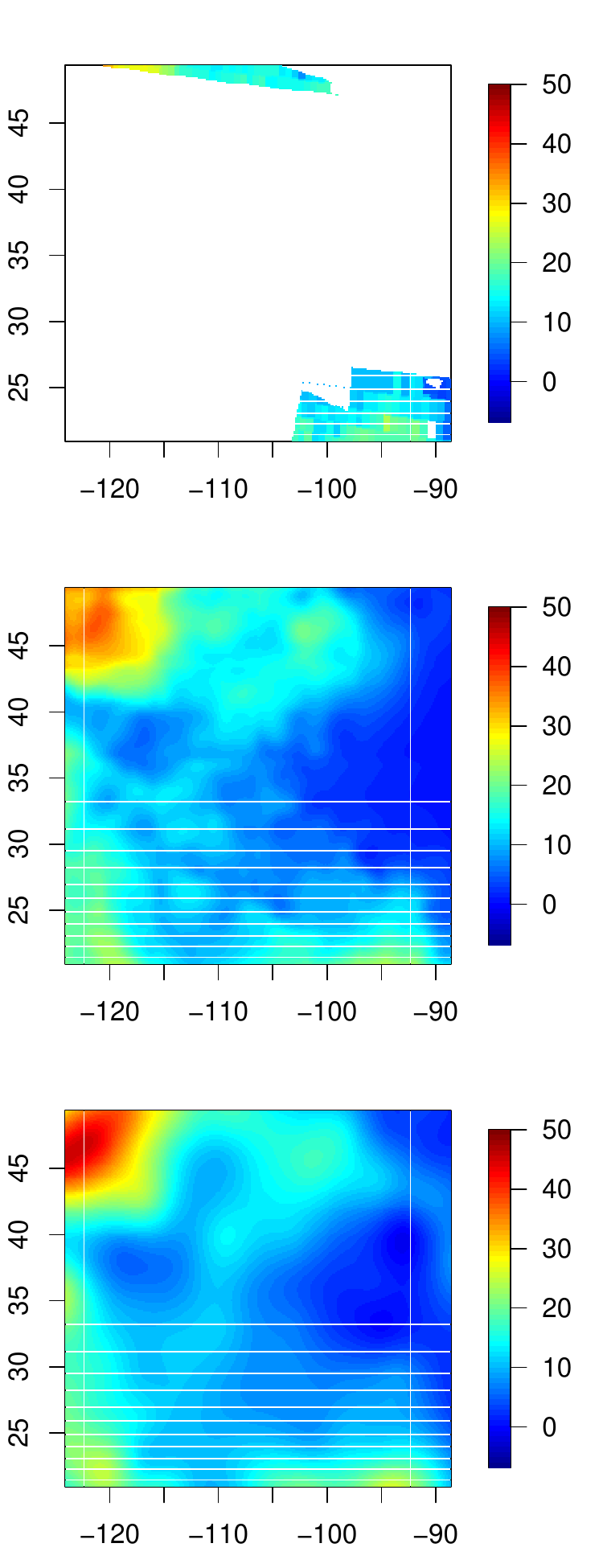}
        \caption*{t = 6}
    \end{subfigure}
    \begin{subfigure}{0.325\textwidth}
        \includegraphics[width=1.0\textwidth]{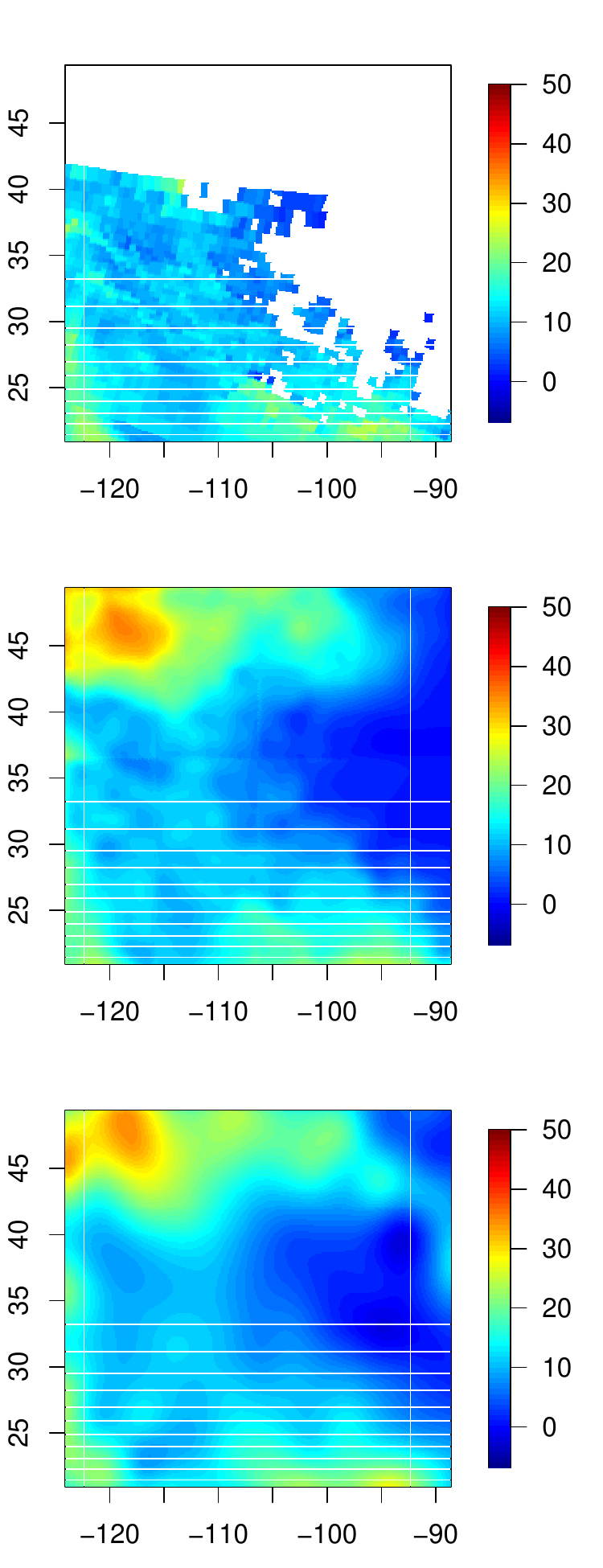}
        \caption*{t = 11}
    \end{subfigure}
    \caption{Total precipitable water measurements (first row), and complete filtering maps made using an HV filter and a LR filter (second and third rows, respectively). The horizontal white stripes are an artifact of the plotting procedure and are not caused by the filtering methods.}
    \label{fig:TPW-filtering}
\end{figure}

We assumed that over the period in which data were collected the dynamics of the TPW can be reasonably approximated by a diffusion process $x$. More formally, and similar to Section \ref{sec:simul-linear}, we assumed that
\begin{equation}
\frac{\partial x}{\partial t} = \alpha \left( \frac{ \partial^2 x }{ \partial^2 s_x } + \frac{ \partial^2 x }{ \partial^2 s_y } \right) + \nu(\bs, t), \label{eq:satdiff}
\end{equation}
where $\nu(\bs, t)$ is a zero-mean stationary Gaussian process with Mat\'ern covariance function with smoothness 1.5, range $\lambda$ and marginal variance $\sigma^2$, independent across time. Exploratory analysis showed that TPW changed slowly between consecutive time points, and so we set $\alpha = 0.99$. We discretized \eqref{eq:satdiff} using centered finite differences with 4 internal steps, which resulted in a model of the form \eqref{eq:linear-dynamics}. Within this context $\bx_t$ represents the values of the process in the middle of each grid cell and $\bfSigma_{0|0} = \frac{1}{1 - \alpha^2}\bQ_1 = \frac{1}{1 - \alpha^2}\bQ_2 = \dots = \frac{1}{1 - \alpha^2}\bQ_{11}$, with $\bfSigma_{0|0}$ obtained by evaluating the covariance function of $\nu$ over the centers of grid cells.

Following \citet{katzfuss2017parallel}, we assumed independent and additive measurement error for each time $t$ and grid cell $i$, such that $y_{ti} \sim \normal(x_{ti}, \tau^2)$ with $\tau = 4.5$. At each time point, we held out randomly selected 10\% of all observations, which we later used to assess the performance of our method.
Covariance-function parameters were chosen using a procedure described in Section S6, resulting in the values $\lambda=2.09$ and $\sigma^2=0.3376$.
We then used Algorithm \ref{alg:eKVLfilter} to obtain the filtering distribution of $\bx_t$ at all time points. 
We applied two approximation methods: the HV filter and the LR filter, which were described in Section \ref{sec:simul-general}. We used the same size of the conditioning set, $N=78$, for both methods, which ensures comparable computation time; for HV, we set $J=4$, $M=7$, and $|\mathcal{X}_\jm|$ equal to 25, 15, 15, 5, 5, 5, 3, 3, respectively, for $m = 0, 1, \dots, M$. The computational cost of the dense Laplace (DL) method is prohibitive for a data set as big as ours.

We compared the performance of the two filters by calculating the root mean squared prediction error (RMSPE) between the filtering means in $\bfmu_{t|t}$ obtained with either method and the true observed values, both only at the held-out test locations. Figure \ref{fig:data-application-scores} shows the values of RMSPE at all time points. We conclude that the HV filter significantly outperformed the LR filter due to its ability to capture much more fine scale detail. Plots of data and filtering results at selected time points are shown in Figure \ref{fig:TPW-filtering} and are representative of full results. We show the plots of the point-wise standard errors in Section S6 in the Supplement.

%\FloatBarrier

%%%%%%%%%%%%%%%%%%%%%%%%%%%%%%%%%%%%%%%%%%%%%%%%%%%%%%%
\section{Conclusions\label{sec:conclusion}}

We specified the relationship between ordered conditional independence and sparse (inverse) Cholesky factors. Next, we described the HV approximation and showed that it exhibits equivalent sparsity in the Cholesky factors of both its precision and covariance matrices. Due to this remarkable property, the approximation is suitable for high-dimensional spatio-temporal filtering.
The HV approximation can be computed using a simple and fast incomplete Cholesky decomposition. Further, by combining the approach with a Laplace approximation and the extended Kalman filter, we obtained scalable filters for non-Gaussian and non-linear spatio-temporal state-space models.

Our methods can also be directly applied to spatio-temporal point patterns modeled using log-Gaussian Cox processes, which can be viewed as Poisson data after discretization of the spatial domain, resulting in accurate Vecchia-Laplace-type approximations \citep{Zilber2019}.
We plan on investigating an extension of our methods to retrospective smoothing over a fixed time period. Another interesting extension would be to combine our methodology with the unscented Kalman filter \citep[][]{Julier1997} for strongly nonlinear evolution.
Finally, while we focused our attention on spatio-temporal data, our work can be extended to other applications, as long as a sensible hierarchical partitioning of the state vector can be obtained as in Section \ref{sec:hv-intro}.

Code implementing our methods and reproducing our numerical experiments is available at \url{https://github.com/katzfuss-group/vecchiaFilter}.

%%%%%%%%%%%%%%%%%%%%%%%%%%%%%%%%%%%%%%%%%%%%%%%%%%%%%%%
\footnotesize
\appendix
\section*{Acknowledgments}

The authors were partially supported by National Science Foundation (NSF) Grant DMS--1654083. MK's research was also partially supported by NSF Grants DMS--1953005 and CCF--1934904, and by the National Aeronautics and Space Administration (80NM0018F0527). We would like to thank Yang Ni, Florian Sch\"afer, and Mohsen Pourahmadi for helpful comments and discussions. We are also grateful to Edward Ott and Seung-Jong Baek for allowing us to use their implementation of the Lorenz model, and to Phil Taffet for advice regarding tailoring this implementation to our needs.

\FloatBarrier

%%%%%%%%%%%%%%%%%%%%%%%%%%%%%%%%%%%%%%%%%%%%%%%%%%%%%%%
%%%%%%%%%%%%%%%%%%%%%%%%%%%%%%%%%%%%%%%%%%%%%%%%%%%%%%%

%%%%
\section{Glossary of graph theory terms\label{app:graph}}

We briefly review here some graph terminology necessary for our exposition and proofs, following \citet{lauritzen1996graphical}.

If $a \rightarrow b$ then we say that $a$ is a \emph{parent} of $b$ and, conversely, that $b$ is a \emph{child} of $a$. Moreover, if there is a sequence of distinct vertices $h_1, \dots, h_k$ such that $h_i \rightarrow h_{i+1}$ or $h_i \leftarrow h_{i+1}$ for all $i<k$, then we say that $h_1, \dots, h_k$ is a \emph{path}. If all arrows point to the right, we say that $h_k$ is a \emph{descendant} of $h_i$ for $i<k$, while each $h_i$ is an \emph{ancestor} of $h_k$.

A \emph{moral graph} is an undirected graph obtained from a DAG by first finding the pairs of parents of a common child that are not connected, adding an edge between them, and then by removing the directionality of all edges. If no edges need to be added to a DAG to make it moral, we call it a \emph{perfect graph}.

Let $G=(V,E)$ be a directed graph with vertices $V$ and edges $E$. If $V_1$ is a subset of $V$, then the \emph{ancestral set} of $V_1$, denoted $\An(V_1)$, is the smallest subset of $V$ that contains $V_1$ and such that for each $v \in \An(V_1)$ all ancestors of $v$ are also in $\An(V_1)$. 

Finally, consider three disjoint sets of vertices $A, B, C$ in an undirected graph. We say that $C$ \emph{separates} $A$ and $B$ if, for every pair of vertices $a \in A$ and $b \in B$, every path connecting $a$ and $b$ passes through $C$.

%%%%%%%%%%%%%%%%%%%%%%%%%%%%%%%%%%%%%%%%%%%%%%%%%%%%%%%
\section{Proofs \label{app:proofs}}

\begin{proof}[Proof of Remark \ref{remark:cond-indep}]
~~ \\ \vspace{-5mm}
    \begin{enumerate}
        \item This proof is based on \citet[][Sect.~3.2]{Schafer2017}. Split $\bw = (w_1,\ldots,w_n)^\top$ into two vectors, $\bu=\bw_{1:j-1}$ and $\bv=\bw_{j:n}$. Then,
\begin{eqnarray*}
\bL & = & \bK^{\frac{1}{2}} = \text{chol}\left( \begin{array}{cc} \bK_{uu} & \bK_{uv} \\ \bK_{vu} & \bK_{vv} \end{array}\right) \\
 & = & \text{chol}\left( \left(\begin{array}{cc} \bI & \bfzero \\ \bK_{vu}\bK_{uu}^{-1} & \bI \end{array}\right) \left(\begin{array}{cc} \bK_{uu} & \bfzero \\ \bfzero & \bK_{vv} - \bK_{vu}\bK_{uu}^{-1}\bK_{uv} \end{array}\right) \left(\begin{array}{cc} \bI & \bK_{uu}^{-1}\bK_{uv} \\ \bfzero & \bI \end{array}\right)  \right) \\
& = & \left(\begin{array}{cc} \bK_{uu}^\frac{1}{2} & \bfzero \\ 
\bK_{vu}\bK_{uu}^{-\frac{1}{2}} & \left( \bK_{vv} - \bK_{vu}\bK^{-1}_{uu} \bK_{vu}\right)^\frac{1}{2}\end{array}\right). \label{eq:blockchol}
\end{eqnarray*}
Note that $\bL_{i,j}$ is the $(i-j+1)$-th element in the first column of $\left( \bK_{vv} - \bK_{vu}\bK^{-1}_{uu} \bK_{vu}\right)^\frac{1}{2}$, which is the Cholesky factor of $\var(\bv|\bu) = \bK_{vv} - \bK_{vu}\bK^{-1}_{uu} \bK_{vu}$. Careful examination of the Cholesky factorization of a generic matrix $\bA$, which is described in Algorithm \ref{alg:ic0} when setting $s_{i,j}=0$ for all $(i,j)$, shows that the computations applied to the first column of $\bA$ are fairly simple. In particular, this implies that 
\[
\bL_{i,j} = \big(\chol(\var(\bv|\bu))\big)_{i-j+1,1} =  \frac{\cov(w_i,w_j|\bw_{1:j-1})}{\sqrt{\var(w_j|\bw_{1:j-1})}},
\]
because $w_j = \bv_1$ and $w_i = \bv_{i-j+1}$. Thus, $\bL_{i,j}=0 \iff \cov(w_i,w_j|\bw_{1:j-1})=0 \iff w_i \perp w_j \, | \, \bw_{1:j-1}$ because $\bw$ was assumed to be jointly normal.
    \item Thm.~12.5 in \citet{Rue2010} implies that for a Cholesky factor $\breve{\bU}$ of a precision matrix $\bP\bK^{-1}\bP$ of a normal random vector $\breve{\bw} = \bP\bw$, we have $\breve{\bU}_{i,j}=0 \iff \breve{w}_i \perp \breve{w}_j \,|\, \{\breve{\bw}_{j+1:i-1}, \breve{\bw}_{i+1:n}\}$. Equivalently, because $\bU = \bP\breve{\bU}\bP$, we conclude that $\bU_{j,i}=0 \iff w_i \perp w_j \,|\, \{\bw_{1:j-1}, \bw_{j+1:i}\}$.
    \end{enumerate}
\end{proof}

%%%%%%%%%%%%%%

\begin{proof}[Proof of Proposition \ref{prop:same-sparsity-factors}]

The fact that $\hat p(\bx)$ is jointly normal holds for any Vecchia approximation \citep[e.g.,][Prop.~1]{Datta2016,Katzfuss2017a}.
By straightforward extension of the proof of \citet[][Prop.~1]{Katzfuss2017a} to the case with nonzero mean $\bfmu$, it can be shown that the Vecchia approximation does not modify the mean (only the covariance matrix). 
Hence, we can write $\hat p(\bx) = \normal_n(\bx|\bfmu,\hat\bfSigma)$.
\begin{enumerate}
\item First, note that $\bL$ and $\bU$ are lower- and upper-triangular matrices, respectively. Hence, we assume in the following that $j < i$ but $x_j \not \in \condset_i$, and then show the appropriate conditional-independence results. 
    \begin{enumerate}
        \item
By Remark \ref{remark:cond-indep}, we only need to show that $x_i \perp x_j | \bx_{1:j-1}$. Let $G$ be the graph corresponding to factorization \eqref{eq:vecchia-approx} and denote by $G_{\An(A)}^m$ the moral graph of the ancestral set of $A$. By Corollary 3.23 in \citet{lauritzen1996graphical}, it is enough to show that $\left\{x_1, \dots, x_{j-1}\right\}$ separates $x_i$ and $x_j$ in $G_{An\left(\left\{x_1, \dots, x_j, x_i\right\}\right)}^m$. In the rest of the proof we label each vertex by its index, to simplify notation.

We make three observations which can be easily verified. [1] $\An\left(\left\{1, \dots, j, i\right\}\right) \subset \left\{1, \dots, i\right\}$; [2] Given the ordering of variables described in Section \ref{sec:hv-intro}, if $k \rightarrow l$ then $k<l$; [3] $G$ is a perfect graph, so $G_{An\left(\left\{1, \dots, j, i\right\}\right)}^m$ is a subgraph of $G$ after all edges are turned into undirected ones.

We now prove Proposition \ref{prop:same-sparsity-factors}.1.a by contradiction. Assume that $\left\{x_1, \dots, x_{j-1}\right\}$ does not separate $x_i$ and $x_j$, which means that there exists a path $(h_1, \dots, h_k)$ in $\left\{x_1, \dots, x_i\right\}$ connecting $x_i$ and $x_j$ such that $h_k \in \An(\left\{1, \dots, j, i\right\})$ and $j+1 \leq h_k \leq i-1$.

There are four cases we need to consider and we show that each one of them leads to a contradiction. First, assume that the last edge in the path is $h_k \rightarrow j$. This violates observation [2]. Second, assume that the first edge is $i \leftarrow h_1$. But because of [1] we know that $h_1<i$, and by [2] we get a contradiction again. Third, let the path be of the form $i \rightarrow h_1 \leftarrow \dots \leftarrow h_k \leftarrow j$ (i.e., all edges are of the form $h_r \leftarrow h_{r+1}$). However, this would mean that $\sx_\jl \subset \sa_\im$, for $x_i\in \sx_\im$ and $x_j\in \sx_\jl$. This implies that $j \in \condset_i$, which in turn contradicts the assumption of the proposition. Finally, the only possibility we have not excluded yet is a path such that $i\leftarrow h_1 \dots h_k \leftarrow j$ with some edges of the form $h_r \rightarrow h_{r+1}$. Consider the largest $r$ for which this is true. Then by [3] there has to exist an edge $h_r \leftarrow h_p$ where $h_p \in \{h_{r+2}, \dots, h_k, j\}$. But this means that $j$ is an ancestor of $h_r$ so the path can be reduced to $i \rightarrow h_1, \dots h_r \rightarrow j$. We continue in this way for each edge "$\leftarrow$" which reduces this path to case 3 and leads to a contradiction. 

Thus we showed that all paths in $G_{An\left(\left\{1, \dots, j, i\right\}\right)}^m$ connecting $i$ and $j$ necessarily have be contained in $\left\{1, \dots, j-1\right\}$, which proves Proposition \ref{prop:same-sparsity-factors}.1.a.

    \item 
Like in part (a), we note that by Remark \ref{remark:cond-indep} it is enough to show that $x_i \perp x_j | \bx_{1:j-1, j+1:i-1}$. Therefore, proceeding in a way similar to the previous case, we need to show that $\{1, \dots, j-1, j+1, \dots, i\}$ separates $i$ and $j$ in $G_{An\left(\left\{1, \dots, i\right\}\right)}^m$. 
    However, notice that it can be easily verified that $\An(\{1, \dots, i\}) \subset \{1, \dots, i\})$, which means that $\An(\{1, \dots, i\}) = \{1, \dots, i\})$. Moreover, observe that the subgraph of $G$ generated by $\An(\{1, \dots, i\})$ is already moral, which means that if two vertices did not have a connecting edge in the original DAG, they also do not share an edge in $G^m_{1, \dots, i}$. Thus $i$ and $j$ are separated by $\{1, \dots, j-1, j+1, \dots, i-1\}$ in $G^m_{1, \dots, i}$, which by Corollary 3.23 in \citet{lauritzen1996graphical} proves part (b).

\end{enumerate}

\item 
Let $\bP$ be the reverse-ordering permutation matrix.
Let $\bB = \chol(\bP\hat{\bfSigma}^{-1}\bP)$. Then $\bU = \bP\bB\bP$. By the definition of $\bB$, we know that $\bB\bB^\top = \bP\hat{\bfSigma}^{-1}\bP$, and consequently $\bP\bB\bB^\top\bP = \hat{\bfSigma}^{-1}$. Therefore, $\hat{\bfSigma} = \left(\bP\bB\bB'\bP\right)^{-1}$. However, we have $\bP\bP = \bI$ and $\bP = \bP^\top$, and hence $
\hat{\bfSigma} = \left((\bP\bB\bP)(\bP\bB^\top\bP)\right)^{-1}$. So we conclude that $\hat{\bfSigma} = (\bU\bU^\top)^{-1} = (\bU^\top)^{-1}\bU^{-1} = (\bU^{-1})^\top\bU^{-1}$ and $(\bU^{-1})^\top = \bL$, or alternatively $\bL^{-\top} = \bU$.
\end{enumerate}
\end{proof}

%%%%%%%%%%%%%%

\begin{proof}[Proof of Proposition \ref{prop:posterior}]
~~ \\ \vspace{-5mm}
\begin{enumerate}
\item 
We observe that HV satisfies the sparse general Vecchia requirement specified in \citet[][Sect.~4]{Katzfuss2017a}, because the nested ancestor sets imply that $\condset_j \subset \condset_i$ for all $j<i$ with $i,j \in \condset_k$.
Hence, reasoning presented in \citet[][proof of Prop.~6]{Katzfuss2017a} allows us to conclude that $\widetilde{\bU}_{j,i}= 0$ if $\bU_{j,i}= 0$.

\item 

%We prove this property by applying Remark \ref{clm:nnz} to the conditional distribution $\left[ \bx | \by \right]$. Then it is enough to show that if $x_j \not\in \condset_i$ then $x_j \perp x_i | \bx_{1:j-1}, \by$. 

As the observations in $\by$ are conditionally independent given $\bx$, we have
\begin{equation}
\textstyle \hat{p}(\bx|\by) \propto \hat{p}(\bx) p(\by|\bx) = \big( \prod_{i=1}^n p(x_i|\condset_i) \big) \big( \prod_{i \in \obs} p(y_i|x_i) \big),
\label{eq:posterior-factorization}
\end{equation}

Let $G$ be the graph representing factorization \eqref{eq:vecchia-approx}, and let $\tilde{G}$ be the DAG corresponding to \eqref{eq:posterior-factorization}. We order vertices in $\tilde{G}$ such that vertices corresponding to $\by$ have numbers $n+1, n+2, \dots, n+|\mathcal{I}|$. For easier notation we also define $\tilde{\obs} = \{i+n : i\in \obs\}$.
% While the proposition still holds if the so-called IW ordering is used instead (see \citet{Zilber2019}), this particular ordering simplifies notation because we can use consecutive integers to label the vertices corresponding to elements of $\bx$. 
Similar to the proof of Proposition \ref{prop:same-sparsity-factors}.1, we suppress the names of variables and use the numbers of vertices instead (i.e., we refer to the vertex $x_k$ as $k$ and $y_j$ as $n+j$). Using this notation, and following the proof of Proposition 1, it is enough to show that $\{1, \dots, j-1\} \cup \tilde{\obs}$ separate $i$ and $j$ in $\tilde{\mathcal{G}}^m$, where  $\tilde{\mathcal{G}} \colonequals G_{\An(\{1, \dots, j, i\} \cup \tilde{\obs})}$.

We first show that $1, \dots, j-1$ separate $i$ and $j$ in $G$. Assume the opposite, that there exists a path $(h_1, \dots, h_k)$ in $\{j+1, \dots, i-1, i+1, \dots, n\}$ connecting $x_i$ and $x_j$. Let us start with two observations. First, note that the last arrow has to go toward $j$ (i.e., $h_k \rightarrow j$), because $h_k>j$. Second, let $p_0 = \max\{p < k: h_p \rightarrow h_{p+1}\}$, the index of the last vertex with an arrow pointing toward $j$ that is not $h_k$. If $p_0$ exists, then $(h_1, \dots, h_{p_0})$ is also a path connecting $i$ and $j$. This is because $h_{p_0}$ and $j$ are parents of $h_{p_0+1}$, and so $h_{p_0} \rightarrow j$, because G is perfect and $h_p > j$.

Now notice that a path $(h_1)$ (i.e., one consisting of a single vertex) cannot exist, because we would either have $i \rightarrow h_1 \rightarrow j$ or $i \leftarrow h_1 \leftarrow j$. The first case implies that $i \rightarrow j$, because in $G$ a node is a direct parent of all its descendants. Similarly in the second case, because $G$ is perfect and $j<i$, we also have that $i \leftarrow j$. In either case the assumption $j \not\in \condset_i$ is violated.

Now consider the general case of a path $(h_1, \dots, h_k)$ and recall that by observation 1 also  $h_k \leftarrow j$. But then the path $(h_1)$ also exists because by \ref{eq:vecchia-approx} all descendants are also direct children of their ancestors. As shown in the previous paragraph, we thus have a contradiction.

Finally, consider the remaining case such that $p_0 = \max\{p:h_p \rightarrow h_{p+1}\}$ exists. But then $(h_1, \dots, h_{p_0})$ is also a path connecting $i$ and $j$. 
If all arrows in this reduced paths are to the left, we already showed it produces a contradiction. If not, we can again find $\max\{p: h_p \rightarrow h_{p+1}\}$ and continue the reduction until all arrows are in the same direction, which leads to a contradiction again. 

Thus we show that $i$ and $j$ are separated by $\{1, \dots, j-1\}$ in $G$. This implies that they are separated by this set in every subgraph of $G$ that contains vertices $\{1, \dots, j, i\}$ and in particular in $\mathcal{G} \colonequals G_{An(\{1, \dots, j-1\} \cup \{j\}\cup \{i\}\cup \tilde{\obs})}$. Recall that we showed in Proposition 1 that $G$ is perfect, which means that $\mathcal{G} = \mathcal{G}^m$.

Next, for a directed graph $\mathcal{F} = (V, E)$ define the operation of \emph{adding a child} as extending $\mathcal{F}$ to $\tilde{\mathcal{F}} = \left(V \cup \{w\}, E \cup \{v \rightarrow w\}\right)$ where $v \in V$. In other words, we add one vertex and one edge such that one of the old vertices is a parent and the new vertex is a child. Note that a perfect graph with an added child is still perfect. Moreover, because the new vertex is connected to only a single existing one, adding a child does not create any new connections between the old vertices. It follows that if $C$ separates $A$ and $B$ in $\mathcal{F}$, then $C \cup \{w\}$ does so in $\bar{\mathcal{F}}$ as well.

Finally, notice that $\tilde{\mathcal{G}}$, the graph we are ultimately interested in, can be obtained from $\mathcal{G}$ using a series of child additions. Because these operations preserve separation even after adding the child to the separating set, we conclude that $i$ and $j$ are separated by $\{1, \dots, j-1\} \cup \tilde{\obs}$ in $\tilde{\mathcal{G}}$. Moreover, because $\mathcal{G}$ was perfect and because graph perfection is preserved under child addition, we have that $\tilde{\mathcal{G}} = \tilde{\mathcal{G}}^m$. 

\end{enumerate}
\end{proof}

%%%%%%%%%%%%%%

\begin{remark}
Assuming the joint distribution $\hat{p}(\bx)$ as in \eqref{eq:vecchia-approx-ind}, we have $\hat{p}(x_i, x_j)=p(x_i, x_j)$ if $x_j \in \condset_i$; that is, the marginal bivariate distribution of a pair of variables is exact if one of the variables is in the conditioning set of the other.
\label{remark:marginals}
\end{remark}
\begin{proof}
First, consider the case where $x_i, x_j \in \sx_\jm$. Then note that $\hat{p}(\sx_\jm) = \int \hat{p}(\bx) d\bx_{-\sx_\jm}$. Furthermore, notice that given the decomposition \eqref{eq:vecchia-approx} and combining appropriate terms we can write 
$$\hat{p}(\bx)=p\left(\sx_\jm|\sa_\jm\right)p\left(\sa_\jm\right)p\left(\bx_{-\left\{\sx_\jm\cup \sa_\jm\right\}}\right).$$
Using these two observations, we conclude that
$$ \textstyle \hat{p}(\sx_\jm) = \int \hat{p}(\bx) d\bx_{-\sx_\jm} = \int \prod_{k=0}^m p(\sx_\jm|\sa_\jm) d(\sx, \sx_{j_1}, \ldots, \sx_\jmm) = p(\sx_\jm),
$$
which proves that $\hat{p}(x_i, x_j)=p(x_i, x_j)$ if $x_i, x_j \in \sx_\jm$.

Now let $x_i \in \sx_\jm$ and $x_j \in \sx_\jl$ with $\ell < m$, because $x_j \in \condset_i$ implies that $j<i$.
%We also set $\tilde{\bx} = \bx(\bigcup_{k=0}^m \knots_\jm)$ and slightly abusing notation we write $p(\knots_\jm)$ when we mean $p(\bx(\knots_\jm))$. Finally, we set $P=\mathcal{P}_\jl$. 
Then,
\begin{align*}
\hat{p}(\sx_\jm, \sx_\jl) & \textstyle =\int \hat{p}(\bx) d\bx_{-\{\sx_\jm \cup \sx_\jl\}} = \\
   & \textstyle = \int \prod_{k=0}^M \prod_\jk p(\sx_\jk|\sa_\jk) d\bx_{-\{\sx_\jm \cup \sx_\jl\}} \\
   & \textstyle = \int \prod_{k=0}^m p(\sx_\jk|\sa_\jk) d(A_\jl\cup\sx_\jlp \cup \dots \cup \sx_\jmm).
\end{align*}
The second equality uses \eqref{eq:vecchia-approx}, the definition of $\hat{p}$; the last equation is obtained by integrating out $\sx^{0:M}\setminus \bigcup_{k=0}^m \sx_\jk$.
Note that $\sa_\jm = \sa_\jl \cup \bigcup_{k=\ell}^{m-1} \sx_\jk$. Therefore, by Bayes law, for any $k>\ell$:
\begin{align*}
    p(\sx_\jk|\sa_\jk) & \textstyle = p\left(\sx_\jk|\sa_\jl \cup (\sa_\jk\setminus \sa_\jl)\right) \\
   & \textstyle = \frac{p\left(\sa_\jl | \sx_\jk \cup (\sa_\jk\setminus \sa_\jl)\right)p(\sx_\jk|\sa_\jk\setminus\sa_\jl)}{p(\sa_\jl|\sa_\jk \setminus \sa_\jl)} \\ 
   & \textstyle =\frac{p(\sa_\jl|\sa_\jkp \setminus \sa_\jl)p(\sx_\jk|\sa_\jk\setminus \sa_\jl)}{p(\sa_\jl|\sa_\jk \setminus \sa_\jl)} = (*)
\end{align*}
The last equality holds because $\sx_\jm \cup \sa_\jm = \sa_\jmp$. As a consequence
\begin{align*}
    \textstyle \prod_{k=0}^m p(\sx_\jk|\sa_\jk) & \textstyle = \prod_{k=0}^m p(\sx_\jk|\sa_\jk\setminus \sa_\jl) p(\sa_\jl) \\
    & \textstyle = \prod_{k=0}^m p(\sx_\jk \rvert \bigcup_{s=k-1}^\ell \sx_\js) p(\sa_\jl)
\end{align*}
and
\begin{align*}
(*) & \textstyle = \int \prod_{k=0}^m p(\sx_\jm|\sa_\jm) p(\sa_\jl) d(\sa_\jl\cup\sx_\jlp \cup \dots, \cup \sx_\jmm)  \\
 & \textstyle = \int p(\sx_\jm, \sx_\jmm, \dots, \sx_\jl, \sa_\jl) d(\sa_\jl\cup\sx_\jlp \cup \dots, \cup \sx_\jmm) \\ 
 & = p(\sx_\jm, \sx_\jl)    
\end{align*}
This means that $\hat{p}(\sx_\jm, \sx_\jl) = p(\sx_\jm, \sx_\jl)$, or that the marginal distribution of $\sx_\jm$ and $\sx_\jl$ in \eqref{eq:vecchia-approx} is the same as in the true distribution $p$.
Because $p$ is Gaussian, it follows that
$
\hat{p}(x_i, x_j) = p(x_i, x_j).
$
 This ends the proof.

\end{proof}

%%%%%%%%%%%%%%

\begin{proof}[Proof of Proposition \ref{prop:chol-lemma}]
We use $l^\text{inc}_{i,j}$, $l_{i,j}$, $\sigma_{i,j}$, $\hat{\sigma}_{i,j}$ to denote the $(i,j)$-th elements of $\bL^\text{inc} = \ichol(\bfSigma, \bS)$, $\bL = \chol(\hat\bfSigma)$, $\bfSigma$, $\hat{\bfSigma}$, respectively. It can be seen easily in Algorithm \ref{alg:ic0} that $\chol(\bfSigma)=\ichol(\bfSigma, \bS^{\bm{1}})$, where $\bS^{\bm{1}}_{i,j} = 1$ for $i\geq j$ and $0$ otherwise.

We prove that $l^\text{inc}_{i,j}=l_{i,j}$ by induction over the elements of the Cholesky factor, following the order in which they are computed. First, we observe that $l^\text{inc}_{1,1}=l_{1,1}$. Next, consider the computation of the $(i,j)$-th entry, assuming that we have $l^\text{inc}_{k,q}=l_{k,q}$ for all previously computed entries. According to Algorithm \ref{alg:ic0}, we have
\begin{equation}
\textstyle l_{i,j} = \frac{1}{l_{j,j}} \big( \hat{\sigma}_{i,j} - \sum_{k=1}^{j-1}l_{i,k}l_{j,k} \big), \quad\quad l^{inc}_{i,j} = \frac{s_{i,j}}{l_{j,j}} \big( \sigma_{i,j} - \sum_{k=1}^{j-1}l_{i,k}l_{j,k} \big).
\end{equation}
Now, if $s_{i,j}=1 \iff x_j \in \condset_i$, then Remark \ref{remark:marginals} tells us that $\sigma_{i,j} = \hat{\sigma}_{i,j}$, and hence $l_{i,j} = l^\text{inc}_{i,j}$. If $s_{i,j}=0 \iff x_j \not\in \condset_i$, then $l^\text{inc}_{i,j}=0$, and also $l_{i,j}=0$ by Proposition \ref{prop:same-sparsity-factors}.1(a). This completes the proof via induction.
\end{proof}

%%%%%%%%%%%%%%
\begin{remark}
    Let $\bx$ have density $\hat{p}(\bx)$ as in \eqref{eq:vecchia-approx}, and let each conditioning set $\condset_i$ have size at most $N$. Then $\bfLambda$, the precision matrix of $\bx$, has $\mathcal{O}(n\condsetN)$ nonzero elements. Moreover, the columns of $\bU = \rchol(\bfLambda)$ and the rows of $\bL = \chol(\bfLambda^{-1})$ each have at most $\condsetN$ nonzero elements.
    \label{clm:nnz}
\end{remark}

\begin{proof}
    Because the precision matrix is symmetric, it is enough to show that there are only $\mathcal{O}(n\condsetN)$ nonzero elements $\bfLambda_{i,j}$ in the upper triangle (i.e., with $i<j$). Let $x_i \not\in \condset_j$. This means that $x_i \not\rightarrow x_j$ and because by \eqref{eq:vecchia-approx} all edges go from a lower index to a higher index, $x_i$ and $x_j$ are not connected. Moreover because $\sa_\jm \supset \sa_\jmm$ there is also no $x_k$ with $k>\max(i,j)$ such that $x_i \rightarrow x_k$ and $x_j \rightarrow x_k$. Thus, using Proposition 3.2 in \citet{Katzfuss2017a} we conclude that $\bfLambda_{i,j}=0$ for $x_i \not\in \condset_j$. This means that each row $i$ has at most $|\condset_i|\leq \condsetN$ nonzero elements, and the entire lower triangle has $\mathcal{O}(n\condsetN)$ nonzero values.
    %Because the precision matrix is symmetric, it is enough to show that there are only $\mathcal{O}(nN)$ nonzero elements $\bfLambda_{i,j}$ in the lower triangle (i.e., with $i>j$). Notice that for $i,j$ such that $x_j \not\in \condset_i$ and $x_i \in \sx_\im$, Remark \ref{clm:markov-prop} says that $x_i \perp x_j \,|\, \sa_\im$. This means that, in particular, $x_i \perp x_j | \bx_{-i,j}$ so $\bfLambda_{i,j}=0$. Therefore each row $i$ has at most $|\condset_i|\leq N$ nonzero elements, and the entire lower triangle has $\mathcal{O}(nN)$ nonzero values.

    Proposition \ref{prop:same-sparsity-factors} implies that the $i$-th column of $\bU$ has at most as many nonzero entries as there elements in the conditioning set $\condset_i$, which we assumed to be of size at most $\condsetN$. Similarly, the $i$-th row of $\bL$ has at most as many nonzero entries as the number of elements in the conditioning set $\condset_i$.
\end{proof}

\begin{remark}
    Let $\bA$ be an $n \times n$ lower triangular matrix with at most $\condsetN<n$ nonzero elements in each row at known locations. Letting $a_{ij}$ and $\tilde{a}_{ij}$ be the $(i,j)$-th element of $\bA$ and $\bA^{-1}$, respectively, assume that $\tilde{a}_{i,j}=0$ if $a_{i,j}=0$. Then, the cost of calculating $\bA^{-1}$ from $\bA$ is $\mathcal{O}(n\condsetN^2)$.
    \label{clm:triangular-sparse-inverse}
\end{remark}

\begin{proof}
    Notice that calculating $\tilde{\ba}_k$, the $k$th column of $\bA^{-1}$, is equivalent to solving a linear system of the form
    \begin{equation}
        \left[\begin{array}{ccccc}
            a_{11} & 0 & 0 & \dots & 0 \\
            a_{21} & a_{22} & 0 & \dots & 0 \\
            a_{31} & a_{32} & a_{33} & \dots & 0 \\
            \vdots & \vdots & \vdots & \ddots & \vdots \\
            a_{n1} & a_{n2} & a_{n3} & \dots & a_{nn}
        \end{array}\right]
        \left[\begin{array}{c} \tilde{a}_{1k} \\ \tilde{a}_{2k} \\ \tilde{a}_{3k} \\ \vdots \\ \tilde{a}_{nk} \end{array}\right]
        =
        e_k,
    \end{equation}
    where $e_{ik} = 1$ if $k=i$ and 0 otherwise. Using forward substitution, the $i$-th element of $\tilde{\ba}_k$ can be calculated as $\tilde{a}_{ik} = \frac{1}{a_{kk}}\left(e_{ik} - \sum_{j=1}^{i-1}a_{ij}\tilde{a}_{jk}\right)$. This requires $\mathcal{O}(\condsetN)$ time, because our assumptions imply that there are at most $N$ nonzero terms under the summation. Moreover, we also assumed that $\tilde{\ba}_k$ has at most $\condsetN$ nonzero elements at known locations, and so we only need to calculate those. Thus computing $\tilde{\ba}_k$ has $\mathcal{O}(\condsetN^2)$ time complexity. As there are $n$ columns to calculate, this ends the proof.
\end{proof}

%%%%%%%%%%%%%%

\begin{proof}[Proof of Proposition \ref{prop:complexity-vf}]

    Starting with the ichol() procedure in Line 1, obtaining each nonzero element $l_{i,j}$ requires calculating the outer product of previously computed segments of rows $i$ and $j$. Because Remark \ref{clm:nnz} implies that each row of $\bL$ has at most $\condsetN$ nonzero elements, obtaining $l_{i,j}$ is $\mathcal{O}(\condsetN)$. Remark \ref{clm:nnz} also shows that for each $i$, there are at most $\condsetN$ nonzero elements $s_{i,j}=1$, which implies that each row of the incomplete Cholesky factor can be calculated in $\mathcal{O}(\condsetN^2)$. Finally, because the matrix to be decomposed has $n$ rows, the overall cost of the algorithm is $\mathcal{O}(n\condsetN^2)$.
    In Line 2, because $\bL^{-1} = \bU^\top$, Proposition \ref{prop:same-sparsity-factors} tells us exactly which elements of $\bL^{-1}$ need to be calculated (i.e., are non-zero), and that there are only $\condsetN$ of them (Remark \ref{clm:nnz}). Using Remark \ref{clm:triangular-sparse-inverse}, this means that computing $\bL^{-1}$ can be accomplished in $\mathcal{O}(n\condsetN^2)$ time. Analogous reasoning and Proposition 2 allow us to conclude that computing $\widetilde{\bL}$ in Line 5 has the same complexity.
    The cost of Line 3 is dominated by taking the outer product of $\bU$, because $\bH$ and $\bR$ are assumed to have only one non-zero element in each row. However, $\bU\bU^\top$ is by definition equal to the precision matrix of $\bx$ under \eqref{eq:vecchia-approx}. Therefore, by Remark \ref{clm:nnz} there are at most $\mathcal{O}(n\condsetN)$ elements to calculate and each requires multiplication of two rows with at most $N$ nonzero elements. This means that this step can be accomplished in $\mathcal{O}(n\condsetN^2)$ time.
    The most expensive operation in Line 4 is taking the Cholesky factor. However, its cost proportional to the square of the number of nonzero elements in each column \citep[e.g.,][Thm.~2.2]{Toledo2007}, which by Remark \ref{clm:nnz} we know to be $\condsetN$. As there are $n$ columns, this step requires $\mathcal{O}(n\condsetN^2)$ time.
    Finally, the most expensive operation in Line 6 is the multiplication of a vector by matrix $\widetilde\bL$. By Proposition \ref{prop:posterior}, $\widetilde{\bL}$ has the same number of nonzero elements per row as $\bL$, which is at most $\condsetN$ by Remark \ref{clm:nnz}. Thus, multiplication of $\widetilde{\bL}$ and any dense vector can be performed in $\mathcal{O}(n\condsetN)$ time.
    To conclude, each line of Algorithm \ref{alg:VF} can be computed in at most $\order(n\condsetN^2)$ time, and so the total time complexity of the algorithm is also $\mathcal{O}(n\condsetN^2)$.
    
    Regarding memory complexity, notice that by Remarks \ref{clm:nnz} and \ref{prop:same-sparsity-factors}, matrices $\bL$, $\bU$, $\widetilde{\bL}$, $\widetilde{\bU}$, and $\bfLambda$ have $\order(n\condsetN)$ nonzero elements, and \eqref{eq:gobs} implies that matrices $\bH$ and $\bR$ have at most $n$ entries. Further, the incomplete Cholesky decomposition in Line 1 requires only those elements of $\bfSigma$ that correspond to the nonzero elements of $\bS$. Because $\bS$ has at most $\condsetN$ non-zero elements in each row by construction, each of the matrices that are decomposed can be stored using $\mathcal{O}(n\condsetN)$ memory, and so the memory requirement for Algorithm \ref{alg:VF} is $\mathcal{O}(n\condsetN)$.
\end{proof}

% %%%%%%%%%%%%%%%%%%%%%%%%%%%%%%%%%%%%%%%%%%%%%%%%%%%%%%%%%%%%%%%%%%%%%%

% BibTeX users please use one of
\bibliographystyle{apalike}      % basic style, author-year citations
\bibliography{mendeley,additionalrefs}

\end{document}